\documentclass[11pt]{article}
\usepackage{amssymb,amsmath,amsthm,mathtools}
\usepackage{graphicx,fullpage}
\usepackage{longtable}
\usepackage[ruled,linesnumbered]{algorithm2e}

\usepackage[shortlabels]{enumitem}
\setlist{
  listparindent=3em,
  itemindent=\parindent,
  parsep=0pt,
  leftmargin=\parindent
}

\newlist{nested}{enumerate}{5}
\setlist[nested]{
  nosep,    
  noitemsep,
  listparindent=2\parindent,
    leftmargin=2\parindent,
  parsep=0pt
}

\usepackage{caption}
\usepackage{verbatim}
\usepackage{comment}

\usepackage{tikz-cd}

\usepackage{hyperref}
\PassOptionsToPackage{unicode}{hyperref}

\newcommand{\old}[1]{}
\mathchardef\mhyphen="2D 
\newtheorem{theorem}{Theorem}

\newtheorem{lemma}[theorem]{Lemma}
\newtheorem{cor}[theorem]{Corollary}
\newtheorem{obs}[theorem]{Observation}

\newtheorem{dfn}[theorem]{Definition}

\newtheoremstyle{case}{}{}{}{}{}{:}{ }{}
\theoremstyle{case}

\theoremstyle{definition}

\def\G{{\mathcal G}}
\def\H{{\mathcal H}}

\def\R{{\mathcal R}}
\def\S{{\mathcal S}}

\newcommand{\ie} [1] {\textit{i.e.,} #1}

\def\etal{{\it et~al.}\,}

\begin{document}

\title{Equivalence Relations for Computing Permutation Polynomials} 
\author{Sergey Bereg\thanks{
Department of Computer Science,
University of Texas at Dallas,
Box 830688,
Richardson, TX 75083
USA. Research of the first author is supported in part by NSF award CCF-1718994. }
\and Brian Malouf$^*$
\and Linda Morales$^*$
\and Thomas Stanley$^*$
\and I. Hal Sudborough$^*$
\and Alexander Wong$^*$
}

\maketitle

\begin{abstract}
We present a new technique for computing permutation polynomials based on equivalence relations. 
The equivalence relations are defined by expanded normalization operations and new functions that map permutation polynomials (PPs) to other PPs. 
Our expanded normalization applies to almost all PPs, including when the characteristic of the finite field divides the degree of the polynomial.
The equivalence relations make it possible to reduce the size of the space, when doing an exhaustive search.
As a result, we have been able to to compute almost all permutation polynomials of degree $d$ at most 10 over $GF(q)$, where $q$ is at most 97.
We have also been able to compute nPPs of degrees 11 and 12 in a few cases. The techniques apply to arbitrary $q$ and $d$.
In addition, the equivalence relations allow the set all PPs for a given degree and a given field $GF(q)$ to be succinctly described by their representative nPPs. 
We give several tables at the end of the paper 
listing the representative nPPs (\ie the equivalence classes) for several values of $q$ and $d$.  
We also give several new lower bounds for $M(n,D)$, the maximum number of permutations on $n$ symbols with pairwise Hamming distance $D$, mostly derived from our results on PPs.
\end{abstract}

\section{Introduction}
\label{sec:intro}

Let $GF(q)$ denote the finite field over $q=p^m$ elements, where $p$ is prime and $m \ge 1$. The prime $p$ is called the {\em characteristic} of the field.
A polynomial $P(x)$ over $GF(q)$ is a {\em permutation  polynomial} ($PP$) if it permutes the elements of $GF(q)$. Let $N_d(q)$ be the number of PPs of degree $d$ over $GF(q)$. 
Lidl and Mullen \cite{lidl88,lidl93} posed the problem of computing $N_d(q)$ and they  
found the following boundary conditions for $N_d(q)$:
\begin{align*}
N_1(q)&=q(q-1)\\
N_d(q)&=0 \text{ if } d|(q-1) \text{ and } d>1\\
\sum_{d=1}^{q-2} N_d(q)&=q!
\end{align*}

Dickson \cite{dickson1897} characterized all PPs of degree up to 6. Hou \cite{Hou15} gave a survey of recent results about PPs.
Chu, Colbourn, and Dukes \cite{chu2004}, using a table of all PPs of degree at most five given by Lidl and Mullen \cite {lidl88}, counted the number of different PPs of degree at most 5.  
Shallue and Wanless \cite{Shallue-Wanless-pp-13} described those of degree 6. 
Li, Chandler, and Xiang \cite{li10} described PPs of degree 6 and 7 over a field of characteristic 2. 

An exceptional PP is a PP $P(x)$ for $GF(q)$, which, for infinitely many $m$, is also a PP for $GF(q^m)$. 
Fan \cite{Fan19a} obtained a classification of all permutation polynomials of degree 7 over $GF(q)$ for any odd prime power $q>7$ up to linear transformations and proved that there are no non-exceptional PPs of degree 7 for finite fields of order $q$ when $q>49$. 
Fan also \cite{Fan19b} obtained a classification of all permutation polynomials of degree 8 over $GF(q)$ for any odd prime power $q>8$ up to linear transformations and proved that there are no PPs of degree 8 for finite fields of order $q$ when $q>31$. 
In addition, Fan \cite{Fan19} described all PPs of degree 8 over finite fields of characteristic 2 and proved that non-exceptional PP's of degree 8 do not exist over GF($2^r$) if $r > 6$.

It is an interesting problem to enumerate all permutation polynomials of somewhat large degree.
Let $P(x)=a_dx^d+ a_{d-1}x^{d-1}+ \dots  + a_1x+ a_0$ be a degree $d$ permutation polynomial over $GF(q)$.
A brute force search for degree $d$ permutation polynomials over $GF(q)$ 
would require $O(dq^{d+2})$ time as there are $d+1$ coefficients, with $q$ choices for each one, and for each of these possibilities, one needs to examine the list of $q$ values formed by the $d$ terms of the polynomial on each element of $GF(q)$ to see if the result is a permutation. 
One way to make the search more efficient is to look for all {\em normalized} PPs (denoted by {\em nPPs}), which fixes certain coefficients. 
In normalization,
$a_d=1$, $a_0=0$, and when the field characteristic $p$ does not divide $d$, we have $a_{d-1}=0$ \cite{lidl88}. 
It is known that any PP can be transformed into an nPP by certain algebraic operations, which we will describe shortly.
In this paper, we use the names {\em c-normalization} to denote the case where $p \nmid d$, {\em m-normalization} when $p>2 \mid d$, and {\em b-normalization} when $p=2 \mid d$.
Properties of the latter two types of normalization will be investigated in the next section.
In all three cases of normalization, searching for nPPs takes $O(dq^{d-1})$ time as three coefficients are fixed. 

We define transformations that map PPs to PPs, and derive new equivalence relations based these functions, which yield a more succinct  classification of PPs.
Using these functions, we are able to fix a fourth coefficient, making the search for nPPs take $O(dq^{d-2})$ time.
Our new equivalence relations are defined for arbitrary $q$ and $d$.
We provide results for PPs of somewhat large $q$ and $d$.
Specifically, we characterize and count PPs of degree at most 9 for finite fields for primes and prime powers up to $q=97$,
PPs of degree 10 for primes and prime powers up to $q=73$, and PPs of degree 11 for primes and prime powers up to 32. 
This is a near complete accounting of PPs up to the stated limits for $q$ and $d$. 

Let $n$ be a positive integer and let $\sigma$ and $\pi$ be permutations over $n$ symbols. 
A {\em permutation array (PA)} is a set of permutations on $n$ symbols. 
$\sigma$ and $\pi$  have {\em Hamming distance at least D}, denoted by $hd(\sigma,\pi) \ge D$, if $\sigma(x) \ne \pi(x)$ in at least $D$ different positions $x$. 
A PA $A$ has Hamming distance $D$, denoted by $hd(A) \ge D$, if every pair of permutations in $A$ has Hamming distance at least $D$.  
The maximum number of permutations in a PA $A$ on $n$ symbols with $hd(A) \ge D$ is denoted by $M(n,D)$. 

We are interested in an application of permutation polynomials to permutation arrays with large Hamming distance. 
Much recent work has focused on computing large permutation arrays with a given lower bound for their pairwise Hamming distance \cite{bls-18,bmmms-19,bmms-kp-17,bmms17,bmms-19,bms17,chu2004,colbourn2004,gao13,jani15,jani18,Nguyen,smith2012new,wang17,Neri19}. 
Chu \etal \cite{chu2004} showed that PPs can be used for lower bounds on $M(n,D)$. 
Let $P(x)$ and $Q(x)$ be two degree $d$ PPs over  $GF(q)$.
Basically, $P(x)$ and $Q(x)$ can agree in at most $d$ points, because for every set of $d$ points there is one and only one polynomial of degree $d$ that passes through them. 
So, the corresponding permutations must disagree in at least $q-d$ positions. 
That is, the permutations have Hamming distance at least $q-d$. 
So, it follows that, if $\sum_{k=1}^d {N_k(q)} \ge T$, then $M(q,q-d) \ge T$, for some positive integer $T$.
In addition, Chu, \etal \cite{chu2004} showed that when $q$ is a prime power:
\begin{itemize}
    \item
        $M(q,q-d)\ge \sum_{k=1}^d N_k(q).$  
    \item
        If $q=2^k\not\equiv 1\pmod 3$, then $M(q,q-3)\ge (q+2)q(q-1) \text{ and } M(q,q-4)\ge \frac 13 q(q-1)(q^2+3q+8)$.
    \item
        If  $q\not\equiv 2\pmod 3$, then $M(q,q-2)\ge q^2$. 
\end{itemize}
R. Sobhani, \etal \cite{Sobhani2019} computed some values of $N_d(q)$ and used these to give lower bounds for some values of $M(n,D)$. Bereg, \etal \cite{bmms-19} give a table with several new lower bounds for $M(n,D)$ for $n \leq 550$ and a table for $M(n,n-1)$ for prime powers $n \leq 600$.

The paper is organized as follows. In Section \ref{sec:Normalized-PPs}, we discuss normalization of permutation polynomials.
In Sections \ref{subsec:m-normalization} and \ref{subsec:b-normalization}, we extend the concept of normalization by defining m-normalization and b-normalization, to include cases when the degree of the polynomial is a multiple of the field characteristic $p$.
These new forms of normalization apply when $p \mid d$, to which 
classic c-normalization does not apply.
In Section \ref{subsec:m-normalization}, we show that for all $d$ such that $p \mid d$ and $p \ne 2$,
one can limit a search for m-normalized PPs, where $a_d=1$, $a_0 = 0$, and either $a_{d-1} = 0$ or $a_{d-2} = 0$. 
So, again such a search takes $O(dq^{d-1})$ time, as three coefficients are fixed. 
In Section \ref{subsec:b-normalization}, we show that when $p \mid d$, $p=2$ and $d$ is in an interval $2^i \leq d \leq 2^{i+1}$, for some $i$, 
then one can limit a search for b-normalized PP’s, where $a_d=1$, $a_0=0$, and either $a_{e-2}=0$ or $a_{e-1}=0$, for a specified position $e$.
Again, a such search takes $O(dq^{d-1})$ time.

In Section \ref{sec:F-and-G-maps}, we consider other transformations that map nPPs to nPPs, allowing an order of magnitude optimization of the search, as mentioned earlier. 
We call these transformations the $F \mhyphen map$ and the $G \mhyphen map$. 
In Section \ref{sec:F-G equivalence classes}, we use the $F$-map and the $G$-map
to define more inclusive equivalence relations on PPs with larger equivalence classes,
which permit a  more succinct description of all PPs of degree $d$ for a field $GF(q)$.
As the $F$-map allows us to fix another coefficient in a search for nPPs, the time complexity is reduced to $O(dq^{d-2})$.
In Section \ref{sec:algorithms}, we describe algorithms that implement the search for nPPs using equivalence relations.

Finally, in Section \ref{sec:results}, we present our numeric results.
We provide a table listing specific values for $N_{11}(q)$.
We present several tables that list the new categorizations of PPs (using equivalence relations), as well as the number of PPs in each category, for $q$ up to 97 and degrees up to 11.
We summarize all of these results in a table that lists the number of nPPs, equivalence classes, and total number of PPs for $q \le 97$ and degree $d, 6 \le d \le 10$. 
We also include a table giving new lower bounds for $M(n,D)$.

{\em Notation}. 
In this paper, we use $p$ to denote a prime and $q=p^m$ to denote a power of the prime $p$ for some $m \geq 1$.
We use $d$ ($2 \le d<q$) to denote the degree of a PP  $P(x)= a_{d}x^{d} + a_{d-1}x^{d-1} + \dots + a_1x + a_0$ over the finite field $GF(p^m)=GF(q)$ with field characteristic $p$.
Throughout the paper and in the tables at the end we use the following notation for the elements of $GF(q)$. 
$GF(q)$ has $q-1$ non-zero elements, all of which can be listed as $t^0, t^1, \dots , t^{q-2}$, where $t \neq 0$ represents a generator of the multiplicative group of nonzero elements of $GF(q)$.
We use the notation $t^0 = 1, t^1 = 2, \dots , t^{q-2} = q-1$. 
Lidl and Niederreiter \cite{LidlNiederreiter94} give this as the second choice of notations to describe the elements of a finite field, with the first choice for $GF(p^m)$ being by degree m polynomials with coefficients from $GF(p)$. 
PPs can easily be converted from one notation to the other.
As a primitive polynomial is needed to do the appropriate arithmetic, we give explicit primitive polynomials for our results.

\section{Normalized permutation polynomials and equivalence relations} 
\label{sec:Normalized-PPs}

As described in the introduction, a c-normalized nPP is a PP
where $a_d=1$, $a_{d-1}=0$, and $a_0=0$.
This form of normalization
can be achieved when $p \nmid d$. 
We introduce two new definitions for normalization for the case $p \mid d$, which we call {\em m-normalization} and {\em b-normalization}.
Table \ref{tbl-NormTypes} summarizes the three types of normalization.

\begin{table}[h!tb]
\centering
\vspace*{4mm}
\begin{tabular}{|c|c|l|}
\hline
{\bfseries  normalization type} & {\bfseries degree restriction} & {\bfseries nPP properties} \\
\hline
{\em c-normalization} & $p \nmid d$ & monic, $P(0)=0$, $a_{d-1}=0$ \\
{\em m-normalization} & $p \mid d$ and $p > 2$ &monic, $P(0)=0$, either $a_{d-1}=0$ or $a_{d-2}=0$  \\
{\em b-normalization} & $p \mid d$ and $p = 2$ & monic, $P(0)=0$, if $2^i \le d \le 2^{i+1}-3$ for some $i$\\
 & & then either $a_r=0$ or $a_{r-1}=0$, where $r=2^i-1$ \\
\hline
\end{tabular}
\caption{Types of normalization for PPs, $P(x)= a_{d}x^{d} + a_{d-1}x^{d-1} + \dots + a_1x + a_0$, of degree $d$ with field characteristic $p$.}
\label{tbl-NormTypes}
\end{table}

Here are some examples of normalized PPs:
\begin{itemize}
    \item The degree-9 PP $x^9+2x^7+3x^5$ over $GF(5^2)$ is c-normalized,  as $a_9=1$, $a_8=0$ and $a_0=0$.
    \item The degree-6 PP is $x^6+x^5+x^3+5x^2+5x$ over $GF(3^2)$ is m-normalized, as $a_6=1$, $a_4=0$ and $a_0=0$.
    \item The degree-10 PP
    $x^{10}+x^9+x^7+26x^5+30x^4+21x^2+31x$ over $GF(2^5)$ is b-normalized, as
    $a_{10}=1$, $a_6 =0$ and $a_0=0$ (with $i=3$ and $r=2^i-1=7$).
\end{itemize}

The normalization definitions of nPPs allow nearly all PPs of degree $d$ to be converted to nPPs of degree $d$ by applying certain algebraic operations, called {\em normalization operations}, in some order.  
The  normalization operations \cite{lidl88, Shallue-Wanless-pp-13} on a PP $P(x)$ are
\begin{itemize}
    \item multiplication by a constant, \ie $aP(x)$, for some nonzero constant $a$ in $GF(q)$,
    \item addition to the variable, \ie $P(x+b)$, for some constant $b$ in $GF(q)$.
    \item addition of a constant, \ie $P(x)+c$, for some constant $c$ in $GF(q)$, and
\end{itemize}
\noindent We include an  extended normalization operation, namely
\begin{itemize}
    \item multiplication of the variable, \ie $P(sx)$, for some nonzero constant $s$.
\end{itemize}
\noindent Extended normalization will come into play in Section \ref{sec:F-G equivalence classes}. It is customary in the literature to refer to the  normalization operations applied to a PP $P(x)$ by $aP(x+b)+c$, where $a$, $b$, and $c$ are elements of the finite field, with $a$ being nonzero. References to extended normalization applied to a PP $P(x)$, denoted by $aP(sx+b)+c$, where $a, b, c$, and $s$ are elements of the finite field, with both $a$ and $s$ are nonzero, appear, for example, in \cite{Fan19a,Fan19, Fan19b} and have been called linear transformations.

Notice that each operation has an inverse. 
That is, if $e$ is the multiplicative inverse of $a$ in $GF(q)$, then $e(a(P(x))=P(x)$. 
If $h$ is the additive inverse of $b$ in $GF(q)$, then $P((x+b)+h)=P(x)$, etc. 
Let $aP(x+b)+c$ denote the result of performing normalization operations in any order. 
Note that the second operation, addition of a constant $b$ to the variable $x$, when $P(x)$ is a PP, has the effect of permuting the order of its values. 
Note also that, if one adds $c$ and then multiplies by $a$, the result is $aP(x)+ac$. 
Whereas, if we multiply by $a$ and then add $c$, the result is $aP(x)+c$. 
Since we consider adding all such constants, and $a$ times $c$ is a constant in $GF(q)$, the set of such PPs is the same regardless of the order of operations.

We define two equivalence relations, $\R$ and $\R_E$, on PPs of degree $d$, for some $d>1$, as follows.

\begin{dfn} 
\label{def:normalization-relations}
Let $P(x)$ and $Q(x)$ be PPs of degree d, for some d. If $P(x)$ can be converted into $Q(x)$ by some sequence of normalization operations, then $P(x)$ and $Q(x)$ are related by $\R$.
If $P(x)$ can be converted into $Q(x)$ by some sequence of normalization operations including the {\em extended} normalization operation, then $P(x)$ into $Q(x)$ are related by $\R_E$.
\end{dfn}

It is easy to see that $\R$ and $\R_E$ are equivalence relations. 
Observe that any PP $P(x)$ can be transformed to itself by the empty sequence of normalization operations. If $P(x)$ can be transformed into $Q(x)$ by some sequence, then $Q(x)$ can be transformed into $P(x)$ by the inverse of each step in the reversed sequence. 
Finally, if $P(x)$ can be transformed into $Q(x)$, and $Q(x)$ can be transformed into $T(x)$, then $P(x)$ can be transformed into $T(x)$ by the composition of the two sequences.

An nPP is a representative of an equivalence class defined by $\R$.
We can make a more efficient search algorithm by searching for nPPs and the equivalence class they represent, rather than searching directly for PPs. In order to use equivalence classes to count PPs, we need to explore specific properties of normalization. 

\subsection{m-normalization}
\label{subsec:m-normalization}

If the degree $d$ of a PP $P(x)$ is a multiple of the field characteristic $p$, 
then there is no constant $b$ which when added to the variable $x$ will make the coefficient $a_{d-1}$ of $P(x+b)$ equal to $0$. 
This is due to the fact that in this case, the expansion of $(x+b)^d$, does not have a nonzero term involving $x^{d-1}$. 
So, when $p \mid d$, we cannot necessarily achieve c-normalization. 
However, if $a_{d-1} \ne 0$, then there is a constant $b$ such that $P(x+b)$ has $a_{d-2} = 0$.
That is, we can achieve m-normalization. 
This is similar to Fan's idea in Observation 8 \cite{Fan19a}. Recall from Table \ref{tbl-NormTypes} that m-normalization applies when   $p\neq 2$ and $d$ is a multiple of $p$. 
In this section, we show that  m-normalization can be achieved for any PP satisfying these conditions.

As stated in Section \ref{sec:Normalized-PPs}, $P(x)$ is {\em m-normalized} when $a_{d}=1, ~ a_{0}=0$, and either $a_{d-1}=0$ or $a_{d-2}=0$. 
We show that any PP of degree $d$ over $GF(p^m)$ where $p \mid d$ and $p \neq 2 $ can be m-normalized. 
By the definition of characteristic of a finite field, for all $k \in GF(q)$, $ 
\underbrace {k+k+\dots k}_{\text{$d$ times}}=0$. 
Furthermore,
for all $c \in GF(q)$, there exists a $k\in GF(q)$ such that 
$\underbrace{k+k+\dots k}_{\text{$d-1$ times}}=c$.
To see this let  $\underbrace{1+1+\dots 1}_{\text{$d-1$ times}}=s$. Note that $s \ne 0$, since $\underbrace{1+1+\dots 1}_{\text{d times}} = 0$. Since $\underbrace{k+k+\dots k}_{\text{d-1 times}}=ks$, one wants $k(s)=c$ or, equivalently, $k=c/s
$ where the arithmetic is in $GF(q)$.

\begin{theorem}
\label{th:m-normalization}
Any PP P(x) where the degree d is a multiple of the field characteristic p, can be transformed to an m-normalized PP $Q(x)$ by the normalization operations.
\end{theorem}

\begin{proof}
Let $P(x)= a_{d}x^{d} + a_{d-1}x^{d-1} + \dots + a_1x + a_0$, and for some  $a,b,c \in GF(q)$ with $a \neq 0$, let
\begin{align*}
    Q(x) &=aP(x+b)+c \\
    &= aa_d(x+b)^d + aa_{d-1}(x+b)^{d-1} + aa_{d-2}(x+b)^{d-2} + \dots + aa_1(x+b) + aa_0 +c \\
    &= b_dx^d + b_{d-1}x^{d-1} + b_{d-2}x^{d-2} + \dots + b_1x + b_0, 
\end{align*}
Observe that the degree $d$ term of $Q(x)$ has the coefficient $b_d=aa_{d}$.  
If we choose $a$ to be the multiplicative inverse of $a_{d}$, then the degree $d$ coefficient of $Q(x)$ will be 1.

If $a_{d-1}=0$, then $b_{d-1}=0$  and the desired property is true. 
So suppose that $a_{d-1} \neq 0$, and consider $b_{d-2}$ in $Q(x)$. 
Since $d$ is a multiple of $p$, the expansion of $a_{d}(x+b)^{d}$ will derive nonzero coefficients only for terms whose degrees are multiples of $p$.  
Since $p>2$, this means that $(d-2) \nmid p$, so $a_{d}(x+b)^{d}$ will have a coefficient of 0 for the degree $d-2$ term.
Hence $b_{d-2}$ is calculated solely by the expansion of $aa_{d-1}(x+b)^{d-1}+aa_{d-2}(x+b)^{d-2}$. 

The expansion of $aa_{d-1}(x+b)^{d-1}$ will produce a term of degree $d-2$ with coefficient $aa_{d-1}b'$ where $b'=\sum_{1}^{d-1}b$.  
The expansion of $aa_{d-2}(x+b)^{d-2}$ will produce a term of degree $d-2$ with coefficient $b_{d-2}=aa_{d-2}$.  
Therefore the coefficient of $x^{d-2}$ in $Q(x)$ is $b_{d-2}=aa_{d-1}b'+aa_{d-2}=a(a_{d-1}b'+a_{d-2})$.  
And by algebra, $b_{d-2}$ is zero if $a_{d-1}b'+a_{d-2}=0$.  
Since $a_{d-1}\neq 0$ and $d-1$ is not a multiple of $p$, we can choose $b$ such that $b'$ is the additive inverse of $a_{d-2}/a_{d-1}$, making $b_{d-2}=0$ in $Q(x)$.
So in $Q(x),~ b_d=1, b_0=0,$ and either  $b_{d-1}=0$ or $b_{d-2}=0$. 

If we choose $c$ to be the additive inverse of the constant term of $aP(x+b)$, then the constant term becomes zero. So, we achieve m-normalization.

\end{proof}

\subsection {b-normalization}
\label{subsec:b-normalization}

In Section \ref{subsec:m-normalization}, we considered PPs whose degree $d$ is a multiple of the field characteristic $p$ and $p>2$. 
We showed that under those conditions, {\em m-normalization} can be achieved.
In this section, we consider the remaining case, namely, $p \mid d$ and $p=2$, 
and show that, {\em b-normalization} can be achieved except when $d=2^i-2$, for some $i \geq 2$. 

We say that the integer interval $[r,s]$ has a $[t,u]$ gap, if for all $d \in [r,s]$, the expansion of $(x+b)^d$, does not include any nonzero $x^e$ monomials, where $e \in [t,u]$.  For example, the integer interval $[8,13]$ has a $[6,7]$ gap as seen by:

$(x+b)^8 = x^8 + b^8$

$(x+b)^9 = x^9 + bx^8 + b^8x + b^9$

$(x+b)^{10} = x^{10} + b^2x^8 + b^8x^2 + b^{10}$

$(x+b)^{11} = x^{11} + bx^{10} + b^2x^9 + b^3x^8 + b^8x^3 + b^9x^2 + b^{10}x + b^{11}$

$(x+b)^{12} = x^{12} + b^4x^8 + b^8x^4 + b^{12}$

$(x+b)^{13} = x^{13} + bx^{12} + b^4x^9 + b^5x^8 + b^8x^5 + b^9x^4 + b^{12}x + b^{13}$.

\noindent That is, in each of the exhibited expansions there are no $x^6$ or $x^7$ terms.

We use this observation in the proofs below regarding  b-normalization.
In Lemma \ref{lemma:gap}, we show that in the expansion of $(x+b)^d$ when $d \in [2^i,2^{i+1}-3]$, the coefficients of the terms $x^{2^i-2}$ and $x^{2^i-1}$ term are always zero.
That is, the interval $[2^i,2^{i+1}-3]$ has a $[2^i-2,2^i-1]$ gap. 
In Lemma \ref{lemma:b-normalization} we show that  any PP $P(x)$ of characteristic $2$ such that $d$ is a multiple of $2$, has a related PP $P(x+b)$ for which certain terms in its expansion always have a coefficient equal to zero. 
We use Lucas' Theorem (stated below for the reader's convenience) in the proof of Lemma \ref{lemma:gap}.

\begin{theorem}{[Lucas' Theorem\cite{Cameron95} ]}\label{thm:Lucas}
Let $p$ be prime, and let $m=m_0+m_1p+ \dots, +m_dp^d$ and $n=n_0+n_1p+ \dots, +n_dp^d$, where $0 \le m_i,n_i<p$ for $i=0,1,\dots,d$. Then
\[ \binom mn \equiv \prod_{i=0}^{d}{\binom {m_i}{n_i}\pmod{p} }\].
\end{theorem}

\begin{lemma}{[Gap Lemma]}\label{lemma:gap}
 For all $i > 1$, the expansion of $(x+b)^d$, for $d \in [2^i,2^{i+1}-3]$, has a $[2^i-2, 2^i-1]$ gap.
\end{lemma}

\begin{proof}
Consider the expansion $(x+b)^d=\sum_{k=0}^d{\binom dk x^{d-k}b^k}$.
Let  $k \in \{d-(2^i-2),d-(2^i-1)\} $.
Clearly, $ b^k$ is not zero, so our job is to show that the expression $ \binom dk$ is zero. 
Let $k'=d-k$, that is $k' \in \{2^i-2,2^i-1\} $.
By a well-known identity, we have $\binom d{k} =\binom d{d-k} =\binom d{k'}$. 
Represent $d$ and $k'$ by their base-2 $(i+1)$-tuples 
$\delta=(\delta_{i},\delta_{i-1}, \dots,\delta_{2},\delta_{1},\delta_{0})$ and $\kappa=(\kappa_{i},\kappa_{i-1}, \dots,\kappa_{2},\kappa_{1},\kappa_{0})$, respectively, where for all $i,~\delta_i,\kappa_i \in \{0,1\}$. 
Observe that at least one $\delta_j~(0 \le j \le i-2)$ must be 0 because $2^i \le d \le 2^{i+1}-3$.
Observe also that $\kappa_i=1$ for all $i>0$.
Hence there is a $j$ such that $\delta_j=0$ and $\kappa_j=1$, so by Lucas' Theorem, $\binom d{k'}=0=\binom d{k}$.  
It follows that $(x+b)^d$, for $d \in [2^i,2^{i+1}-3]$, has a $[2^i-2, 2^i-1]$ gap.
\end{proof}

\begin{lemma} \label{lemma:b-normalization}
Let $i > 1$.  Let $d \in [2^i, 2^{i+1}-3]$ be even. For any PP $P(x)$ over $GF(2^m)$, where $m>2$, 
there is a constant $b$ in $GF(2^m)$ such that in the 
PP $P(x+b)$, either the $x^{2^i-1}$ term or the $x^{2^i-2}$ term is zero. 

\end{lemma} 

\begin{proof}
By Lemma \ref{lemma:gap}, the interval $[2^i, 2^{i+1}-3]$ has a $[2^i-2, 2^i-1]$ gap. 
Let $P(x) = a_dx^d + a_{d-1}x^{d-1} + a_{d-2}x^{d-2} + \dots + a_1x + a_0$, 
where $d \in [2^i,2^{i+1}-3]$ is even. 
Adding $b$ to the argument gives: $P(x+b) = a_d(x+b)^d + a_{d-1}(x+b)^{d-1} + a_{d-2}(x+b)^{d-2} + \dots + a_1(x+b) + a_0$. 

If $a_{2^i-1}$ is zero there is nothing to prove, 
so suppose $a_{2^i-1}$ is not zero.
Since $[2^i,2^{i+1}-3]$ has a $[2^i-2, 2^i-1]$ gap, each term $(x+b)^s$, for $s \in [2^i, 2^{i+1}-3]$ 
has no $x^e$ term for $e \in [2^i-2,2^i-1]$. 
This means that $a_{2^i-1}(x+b)^{2^i-1}$ and $a_{2^i-2}(x+b)^{2^i-2}$
are the only possible terms whose expansion has a nonzero $x^{2^i-2}$ term. 
By the binomial theorem, $a_{2^i-1}(x+b)^{2^i-1}=a_{2^i-1}x^{2^i-1}+a_{2^i-1}bx^{2^i-2}+ \dots$, and
$a_{2^i-2}(x+b)^{2^i-2}=a_{2^i-2}x^{2^i-2} \dots$, where low order terms are not shown.
Summing these two expansions and isolating the 
$x^{2^i-2}$ term, we solve for the value of $b$ such that $a_{2^i-1}bx^{2^i-2}+a_{2^i-2}x^{2^i-2}=0$.
We see that when  $b=-a_{2^i-1}/a_{2^i-2}$, 
the coefficient of the $x^{2^i-2}$ term of $P(x+b)$ is zero. 
\end{proof}

For example, let $d=2^3$, and let $P(x) = a_8x^8 + a_7x^7 + a_6x^6 + \dots + a_1x + a_0$.
Adding $b$ to the argument gives: 
\[
\begin{split}
  P(x+b) &= a_8(x+b)^8 + a_7(x+b)^7 + a_6(x+b)^6 + \dots + a_1x + a_0 \\
    &= a_8(x^8+b^8)+a_7(x^7+bx^6+b^2x^5+\dots)+a_6(x^6+b^2x^4+\dots)+\dots \\
     &= a_8x^8+a_8b^8+(a_7x^7+a_7bx^6+\dots)+(a_6x^6+a_6b^2x^4+\dots) + \dots
\end{split}
\]
We want to solve for the value of $b$ that makes the coefficient of the $x^6$ term of $P(x+b)$ zero.
So $a_7bx^6+a_6x^6=0$ is satisfied by $b=-a_7/a_6$.

\begin{theorem}
\label{th:b-normalization}
Any PP P(x) over $GF(2^m)$ for some $m>2$, and $2 \mid d$ can be transformed to an b-normalized PP $Q(x)$ by the normalization operations, except when $d=2^i-2$, for some $i \geq 2$. 
\end{theorem}

\begin{proof}
Let $P(x)= a_{d}x^{d} + a_{d-1}x^{d-1} + \dots + a_1x + a_0$, and for some  $a,b,c \in GF(q)$ with $a \neq 0$, let
\begin{align*}
    Q(x) &=aP(x+b)+c \\
    &= aa_d(x+b)^d + aa_{d-1}(x+b)^{d-1} + aa_{d-2}(x+b)^{d-2} + \dots + aa_1(x+b) + aa_0 +c \\
    &= b_dx^d + b_{d-1}x^{d-1} + b_{d-2}x^{d-2} + \dots + b_1x + b_0, 
\end{align*}
Observe that the degree $d$ term of $Q(x)$ has the coefficient $b_d=aa_{d}$.  
If we choose $a$ to be the multiplicative inverse of $a_{d}$, then the degree $d$ coefficient of $Q(x)$ will be 1.
If we choose $c$ to be the additive inverse of the constant term of $aa_0$, then the constant term of $Q(x)$ will be $b_0=aa_0+c=0$.
By Lemma \ref{lemma:b-normalization}, there is a $b$ such that in $Q(x)$, the coefficient of either the degree $2^i-1$ term or degree $2^i-2$ term equal to 0, except when $d=2^i-2$, for some $i \geq 2$. 
Hence $Q(x)$ is b-normalized.

\end{proof}

\subsection{Counting PPs using equivalence classes based on normalization}

When the degree $d$ of a PP is not a multiple of the field characteristic $p$, the PP can be c-normalized.
As shown below, for any such PP $P(x)$, the there is a unique triple $(a,b,c)$ such that $aP(x+b)+c$ is c-normalized. 
Moreover, each equivalence class contains exactly one nPP, and each PP belongs to exactly one equivalence class.
These properties allow us to count, for a given $q$ and $d$, the number of PPs in each equivalence class.

\begin{obs} \label{obs:unique}

Let $P(x)$ be PP where $p \nmid d$.
Then there is a unique triple $(a,b,c)$ such that $aP(x+b)+c$ is c-normalized.
\end{obs}

\begin{proof}
Let $Q(x)=aP(x+b)+c$.
Since $Q(x)$ is c-normalized, the coefficient of its degree-$d$ term is $aa_d=1$.
Hence $a=a_d^{-1}$. 
The degree-$(d-1)$ term of $Q(x)$ is $(aa_{d-1}+db) x^{d-1}=0$, 
so $b=-aa_{d-1}/d$. 
(Note that $d\ne 0$).
Finally, since the constant term of $Q(x)$, namely $Q(0)$, is 0, $c$ is uniquely determined by the constant term of $aP(x+b)$. 
\end{proof}

Let $Q(x)$ be an nPP of degree $d$. 
The equivalence class under the relation $\R$ containing $Q(x)$, denoted by $[Q]$, is the set
\[ [Q]=\{aQ(x+b)+c~|~a,b,c\in GF(q) \text{ and }a\ne 0\}. \]

\begin{lemma} \label{lemma:sizeof-c-classes}
Let $Q(x)$ be an nPP of degree $d<q$ where $p \nmid d$.
Then all $q^2(q-1)$ polynomials in $[Q]$ are different.
\end{lemma}

\begin{proof}
Let $Q(x) = a_{d}x^{d} + a_{d-1}x^{d-1} + \dots + a_1x + a_0$ and assume that polynomials $P(x)=aQ(x+b)+c$ and $P'(x)=a'Q(x+b')+c'$ in $[Q]$ are equal.
Since the degree-$d$ terms of $P(x)$ and $P'(x)$ are equal, $a=a'$. 
Since the degree-$(d-1)$ terms are equal, $a(a_{d-1}+a_ddb)=a'(a_{d-1}+a_ddb')$. 
Then
$a_{d-1}+a_ddb=a_{d-1}+a_ddb'$ and $db=db'$ and $b=b'$ (since $d$ is not a multiple of $p$).
Let $e$ be the lowest degree term of $aQ(x+b)$. Then the lowest degree terms of $P(x)$ and $P'(x)$ 
are $e+c$ and $e+c'$, respectively. Thus, $c=c'$ and the claim follows. 
\end{proof}

Note that Lemma \ref{lemma:sizeof-c-classes} implies that each equivalence class of $\R$ contains one and only one nPP. 
Thus when $d$ is not a multiple of $p$, each equivalence class contains exactly $q^2(q-1)$ members (including the representative nPP). Note that the equivalence classes by definition are disjoint.
If the number of nPPs is $k$, there are $kq^2(q-1)$ PPs.

\section{Mapping nPPs to nPPs}
\label{sec:F-and-G-maps}

We now describe the {\em $F$-map} and the {\em $G$-map}, two new functions that map nPPs to nPPs. 
We will use these functions in Section \ref{sec:F-G equivalence classes} to define new equivalence relations on nPPs whose equivalence classes are unions of equivalence classes of $\R$.

\subsection{The $F$-map}
\label{subsec:Fmap}

The $F$-map is the function that multiplies the degree $(d-k)$ term of $P(x)$ by $t^k$, for all $k$.
The $F$-map allows one additional coefficient to be fixed, resulting in an order of magnitude speedup in the search for PPs.

\begin{dfn} 
\label{def:F-map}
Define the $\boldsymbol{F}${\bf-map} by
\[    F(P(x)) = t^0 a_d x^d + t^1 a_{d-1} x^{d-1} + \dots + t^k a_{d-k} x^{d-k} + \dots + t^{d-1} a_1 x+t^da_0 = \sum_{k=0}^d t^k a_{d-k} x^{d-k}.\]
\end{dfn}

\noindent First, we show that the set of PPs is closed under the
$F$-map.

\begin{lemma} \label{lemma:ClosureUnderFmap}
If $P(x)$ is a PP, then
$F(P(x))$
is a PP.
\end{lemma}

\begin{proof}
We show that $F(P(x))$ is one-to-one.
Observe that, for any nonzero $t^i$,
\begin{align*}
F(P(t^i))&=t^0a_d(t^i)^d + t^1a_{d-1}(t^i)^{d-1} + \dots + t^{d-1}a_1(t^i)^1+t^da_0\\
&=t^d(a_d(t^d)^{i-1} + a_{d-1}(t^{d-1})^{i-1} + \dots + a_1(t^1)^{i-1}+a_0)\\
&=t^d(a_d(t^{i-1})^d + a_{d-1}(t^{i-1})^{d-1} + \dots + a_1(t^{i-1})^1+a_0)\\
&=t^dP(t^{i-1})\\
&=t^dP(t^i/t).
\end{align*}
It follows that $F(P(x))$ is a permutation polynomial, since it is obtained from $P(x)$ by multiplying by the constant $t^d$ and
replacing the argument $x$ by $t^{-1}x$.
\end{proof}

Lemma \ref{lemma:ClosureUnderFmap} provides an alternative formulation for the $F$-map, namely, 
$F(P(x)) = t^dP(t^{-1}x) =t^dP(x/t).$
Let $ F^i(P(x))$ denote the composition of the $F$-map with itself $i$ times, for some $i$. Then
\begin{equation} \label{eqn:alt-F-i-map}
    F^i(P(x)) = t^{di}P(t^{-i}x) =t^{di}P(x/t^i). 
\end{equation}
\noindent 
We will use this formulation in later proofs in this paper.

By definition, the $F$-map multiplies the coefficient $a_d$ by $t^0=1$, so when $P(x)$ is an nPP, $F(P(x))$ is also an nPP.
This is formalized in Corollary \ref{cor:F-mapGenratesNPPs} below, which follows easily from Lemma \ref{lemma:ClosureUnderFmap}.

\begin{cor}
\label{cor:F-mapGenratesNPPs}
If $P(x)$ is an nPP over GF(q), then
$F(P(x))$
is an nPP.
\end{cor}

Observe that the $q-1$ non-zero elements of $GF(q)$ form a cyclic group, $\G_{q-1}$, under multiplication \cite{LidlNiederreiter94}. 
Moreover, for each $k$, there exists an $r,~ (0<r \le q-1)$, such that the iterates, $ t^{k},t^{2k}, \dots, (t^{rk \mod(q-1)}=1)$, form a cyclic subgroup, $\H_{t^{k}}$, of $\G_{q-1}$. 
By Lagrange's theorem, the number of elements in $\H_{t^{k}}$, \ie ord($\H_{t^{k}})$, 
is a divisor of $q-1$.

Consider iterations of the $F$-map, namely, the sequence 
\begin{equation} \label{eqn:F-cycle on P(x)}
    P(x), F(P(x)),  F^2(P(x)), \dots~ ,  F^i(P(x)), \dots~ , 
\end{equation}
where $F^i(P(x))=\sum_{k=0}^d t^{ik}a_{d-k} x^{d-k}$.
For the degree $(d-k)$ term, iterative use of the $F$-map yields the sequence of coefficients
\[ a_{d-k},~t^k  a_{d-k}, ~t^{2k} a_{d-k},~ \dots~,(t^{rk \mod(q-1)}a_{d-k}=a_{d-k}) \]
where the terms are simply the elements of $\H_{t^{k}}$ multiplied by the common factor $a_{d-k}$.
This forms a cycle of length $r \leq q-1$ where $r$ is smallest integer such that $t^{rk \mod(q-1)}a_{d-k}=a_{d-k}$.
\begin{dfn} 
\label{def:$F_k$-map}
Let  $k$ be an integer such that $1 \leq k \le d$.
Define the $\boldsymbol{F_k}{\bf-map}$ by $F_k(x)=xt^k$.
\end{dfn}
\noindent The $F_k$ function gives us a third way to look at the $F$-map, namely, for each $k$, the $F$-map computes $F_k(a_{d-k})=t^ka_{d-k}$.
In other words, $F(P(x)) = a_dx^d + t^1a_{d-1}x^{d-1} + \dots + t^{d-1}a_1x+t^da_0 = a_dx^d + F_1(a_{d-1})x^{d-1} + \dots + F_{d-1}(a_1)x + F_d(a_0)$.
Note that iterations of the $F_k$-map on the element 1 yields the cyclic subgroup $\H_{t^k}$.
We call this sequence of iterations the $\boldsymbol{F_k}${\em{\bf-cycle}}.
Define the length of the $F_k$-cycle to be $ord(\H_{t^k})$. 
In this paper, we are interested in those values of $k$ for which the $(d-k)^{th}$ coefficient of an nPP $P(x)$ is not 0.

Observe that for any PP $P(x)$, there is an integer $s \ge 1$, such that the sequence shown in Equation (\ref{eqn:F-cycle on P(x)}) forms a cycle. 
\begin{dfn} 
\label{def:F-cycle on P(x)}
The sequence of iterates of the $F$-map on the PP $P(x)$, namely
\[
P(x), F(P(x)),  F^2(P(x)), \dots, F^s(P(x)) = P(x) 
\]
is called the  $\boldsymbol{F}$\bf{-cycle on }$\boldsymbol{P(x)}$.
\end{dfn}

Consider the values of $k$ for which the coefficient $a_{d-k}$ in $P(x)$ is nonzero.
For all such $k$, 
$(1 \le k \le d)$, 
let $g_k=gcd(k,q-1)$,
and let $j = \min\limits_{k}\{g_k\}$. 
The length of the $F$-cycle on $P(x)$ is $s=(q-1)/j$.
That is, the length of the $F$-cycle on $P(x)$ is the least common multiple of the orders of the subgroups $\H_{t^{k}}$ for all $k$ such that $a_{d-k} \neq 0$.




For example, consider $GF(5^2)$. 
The cyclic subgroups of $\G_{5^2-1}$ are 
$\H_{t^0}$, $\H_{t^1}$, $\H_{t^2}$, $\H_{t^3}$, $\H_{t^4}$, $\H_{t^6}$, $\H_{t^8}$, $\H_{t^{12}}$, and their orders are 1, 24, 12, 8, 6, 4, 3 and 2, respectively. 
To see this, observe that 
\begin{align*}
    \H_1=\H_{t^0} &=\{t^{i*0}\}=\{1\} \text{ for all } i   &\H_{7}=\H_{t^6} &=\{t^{i*6}\} \text{ for all } i\\
    \H_{2}=\H_{t^1} &=\{t^{i*1}\} \text{ for all } i         & &=\{t^{0*6},t^{1*6},t^{2*6},t^{3*6},t^{4*6}\}\\
        &=\{t^{0*1},t^{1*1},t^{2*1},t^{3*1},\dots\}          & &=\{t^0,t^6, t^{12}, t^{18},t^{24}=t^{0} \}\\
        &=\{t^0,t^1, t^2, t^3,\dots,t^{23},t^{24}=t^{0} \}   & &=\{1,7,13,19\}\\
        &=\{1,2,3,4,\dots,24\}                               & &=\H_{19}\\
        &=\H_{6}=\H_{8}=\H_{12}=\H_{14}=\H_{18}=\H_{20}=\H_{24}   &\H_{9}=\H_{t^8} &=\{t^{i*8}\} \text{ for all } i\\
    \H_{3}=\H_{t^2} &=\{t^{i*2}\} \text{ for all } i         & &=\{t^{0*8},t^{1*8},t^{2*8},t^{3*8}\}\\
        &=\{t^{0*2},t^{1*2},t^{2*2},t^{3*2},\dots\}          & &=\{t^0,t^8, t^{16}, t^{24}=t^{0} \}\\
        &=\{t^0,t^2, t^4, t^6,\dots,t^{22},t^{24}=t^{0} \}   & &=\{1,9,17\}\\
        &=\{1,3,5,7,\dots,23\}                               & &=\H_{17}\\  
        &=\H_{11}=\H_{15}=\H_{23}                   &\H_{13}=\H_{t^{12}} &=\{t^{i*12}\} \text{ for all } i\\
    \H_{4}=\H_{t^3} &=\{t^{i*3}\} \text{ for all } i         & &=\{t^{0*12},t^{1*12},t^{2*12}\}\\
        &=\{t^{0*3},t^{1*3},t^{2*3},t^{3*3},\dots\}          & &=\{t^0,t^{12}, t^{24}=t^{0} \}\\
        &=\{t^0,t^3, t^6, t^9,\dots,t^{21},t^{24}=t^{0} \}   & &=\{1,13\} \\
        &=\{1,4,7,10,\dots,22\} \\                           
        &=\H_{10}=\H_{16}=\H_{22} \\                         
    \H_{5}=\H_{t^4} &=\{t^{i*4}\} \text{ for all } i\\
        &=\{t^{0*4},t^{1*4},t^{2*4},t^{3*4},t^{4*4},t^{5*4},t^{6*4}\}\\
        &=\{t^0,t^4, t^8, t^{12},t^{16},t^{20},t^{24}=t^{0} \}\\
        &=\{1,5,9,13,17,21\} \\        &=\H_{21}       
\end{align*}
 
\noindent where $mod~{24}$ arithmetic is used in the exponents. 

To illustrate the computation of $F$-cycles, consider the nPP $P(x)$ = $x^9+2x^7+12x^5+4x^3+17x$ over $GF(25)$. 
For all $k, (1 \le k \le d)$, the non-zero coefficients are $a_7, a_5, a_3$ and $a_1$, which correspond to $k=2,4,6$, and 8 and $g_k=gcd(k,24)=2,4,6$, and 8, respectively.
We next compute $j = \min\limits_{k}\{g_k\}= \min\{2,4,6,8\}=2$.
So, the length of the $F$-cycle on $P(x)$ is $(q-1)/j=24/2=12$.
That $F^{12}(P(x)) = P(x)$ is easily verified.
Note also that the length of each respective $F_k$-cycle is the order of the subgroup $\H_{t^k}$, which, referring to the list above, is $ord(\H_{t^2})=12, ~ord(\H_{t^4})=6, ~ord(\H_{t^6})=4,$
 and $ord(\H_{t^8})=3$,
and their least common multiple is 12, which is the length of the $F$-cycle on $P(x)$.

In general, for larger degree nPPs over $GF(q)$, there is a $k$ such that $gcd(k,q-1)=1$, so the length of the $F$-cycle  is $q-1$. 
This means that if there is a nPP $a_dx^d+ a_{d-1}x^{d-1}+ \dots + a_1x+ a_0$ such that $a_{d-k} \neq 0$, then there is also one in which $a_{d-k} = 1$. 
This allows us to fix an additional coefficient $a_{d-k}$ to the values in $\{0,1\}$, thus reducing the search time for nPPs to $O(dq^{d-2})$.

\subsection{The $G$-map}
\label{subsec:Gmap}

The $G$-map is the function that raises each coefficient in $P(x)$ to the power $p$.
That is,
\begin{dfn} 
\label{def:G-map}
Define the $\boldsymbol{G}${\bf -map} by
\[G(P(x)) = a_d^p x^d + a_{d-1}^p x^{d-1} + \dots + a_1^p x + a_0^p= \sum_{k=0}^d a_{d-k}^p x^{d-k}.\]
\end{dfn}



\noindent We also consider the function $G$ as a function on the elements of $GF(q)$, which follows from the case the polynomial is of degree 0. 


We show that the set of PPs is closed under the $G$-map.

\begin{lemma} \label{lemma:ClosureUnderGmap}
If $P(x)$ is a PP,
then $G(P(x))$ is a PP.
\end{lemma}

\begin{proof}
Observe that for all $c,d$ in $GF(q)$, $(c+d)^p = c^p + d^p$.
This is the so-called ``Freshman Dream'' idea,
which follows from the binomial theorem and the fact that $\binom p i \equiv 0\pmod p$ when $p$ is prime \cite{LidlNiederreiter94}.
So, 
\begin{align*}
(P(x))^p&=(a_dx^d + a_{d-1}x^{d-1} + \dots + a_1x + a_0)^p\\
&= a_d^p(x^p)^d + a_{d-1}^p(x^p)^{d-1} + \dots + a_1^px^p + a_0^p\\
&=G(P(x^p)).
\end{align*}

Also observe that if $(z_0, z_1, \dots , z_{q-1})$ is a permutation of $GF(q)$, 
then so is $(z_0^p, z_1^p, \dots , z_{q-1}^p)$.
Suppose not. 
That is, suppose that $z_i^p = z_j^p$.
Then $(z_i^p - z_j^p)=0$, 
so, $(z_i - z_j)^p =0$. 
Hence, $z_i - z_j = 0$, and so, $z_i=z_j$, a contradiction. 

By assumption, $P(x)$ is a PP, so  $\theta=(P(0),P(1), \dots ,P(q-1))$ is a permutation. 
It follows that $\pi=(P(0)^p, P(1)^p, \dots , P(q-1)^p)$ is also a permutation. 
Since $\pi$ is the permutation generated by the polynomial $(P(x))^p$, it follows that  $(P(x))^p$ is a PP.
Hence, $G(P(x^p))$ is a PP, since, as shown above, $(P(x))^p=G(P(x^p))$. 
This means the permutation generated by $G(P(x^p))$, namely, 
$(G(P(0^p)), G(P(1^p)), \dots,G(P((q-1)^p)))$ is identical to the permutation $\pi$.
That is, $\pi=(P(0)^p, P(1)^p, \dots , P(q-1)^p)=(G(P(0^p)), G(P(1^p)), \dots,G(P((q-1)^p)))$.


Note also that $\sigma=(0^p, 1^p, \dots, (q-1)^p)$ is a permutation.
Applying $P(x)$ to $\sigma$ yields the permutation   $\rho=(P(0^p),P(1^p), \dots ,P((q-1)^p))$. 
Since $\rho$ is a permutation, it is simply a reordering of the permutation  $\theta=(P(0),P(1), \dots ,P(q-1))$.
Hence the sequence $\tau=(G(P(0)), G(P(1)), \dots$, $G(P(q-1)))$ is simply a reordering of the permutation
$\pi=(G(P(0^p)), G(P(1^p)), \dots,G(P((q-1)^p)))$.
That is, $\tau$ is a permutation.
Finally, observe that $\tau$ is the permutation generated by applying the G-map to the PP $P(x)$.
So it follows that $G(P(x)) = (a_d)^px^d + (a_{d-1})^px^{d-1} + \dots + (a_1)^px+ (a_0)^p$ is a PP over $GF(q)$.
\end{proof}

We now show that the set of nPPs is closed under the $G$-map.

\begin{cor}
\label{cor:G-mapGeneratesNPPs}
 If $P(x)$ is a nPP then $G(P(x))$ is also an nPP. 
\end{cor}

\begin{proof}
By Lemma \ref{lemma:ClosureUnderGmap}, we only need to show that $G(P(x))$ is normalized when $P(x)$ is normalized. 
Note that the coefficient of $x^d$ in $P(x)$ is $a_d=1$. Hence, by definition of the $G$-map, the coefficient of $x^d$ in $G(P(x))$ is $a_d^p=1^p=1$. 
Also, since the coefficient of $x^0$ in $P(x)$ is $a_0=0$, then the coefficient of $x^0$ in $G(P(x))$ is $a_0^p=0^p=0$. 
For c-normalization, since the coefficient of $x^{d-1}$ in $P(x)$ is $a_{d-1}=0$, then the coefficient of $x^{d-1}$ in $G(P(x))$ is $a_{d-1}^p=0^p=0$.
For m-normalization, if the coefficient of $x^{d-2}$ in $P(x)$ is $a_{d-2}=0$, then the coefficient of $x^{d-2}$ in $G(P(x))$ is $a_{d-2}^p=0^p=0$.
For b-normalization, if for some specified $j$, $a_{d-j}=0$, then the corresponding coefficient $a_{d-j}^p=0^p=0$ as well.

\end{proof}


By definition, the $G$-map raises each coefficient in the PP $P(x)$  to the power $p$. 
Consider iterations of the $G$-map, namely, the sequence 
\begin{equation} \label{eqn:G-cycle on P(x)}
    P(x), G(P(x)),  G^2(P(x)), \dots~ ,  G^i(P(x)), \dots~ , 
\end{equation}
\noindent where $G^i(P(x))=\sum_{k=0}^d a_{d-k}^{p^i} x^{d-k}$.
Iterative use of the $G$-map on the $(d-k)^{th}$ coefficient yields the sequence 
\[    a_{d-k} = a^{p^0}_{d-k},~ a^{p^1}_{d-k},~ a^{p^2}_{d-k},~ \dots,  ~a^{p^m\pmod {p^m-1}}_{d-k}  =a_{d-k}
\]    

\noindent where  $ ~a^{p^m \pmod {p^m-1}}_{d-k}= a_{d-k}$, because $p^m =1 \pmod {p^m-1}$. 
This forms a cycle of length $m$. 
Let $r_{a_{d-k}}\pmod{p^m-1}$ be the smallest integer such that $a_{d-k}^{p^{(r_{a_{d-k}})}}=a_{d-k}$.
The sequence of coefficients     $a_{d-k},~ a^{p^1}_{d-k},~  \dots,  ~a^{p^{(r_{a_{d-k}})}}_{d-k}
    =a_{d-k}$
is called the  $\boldsymbol{G}$\textbf{\em-cycle
on the coefficient}  $\boldsymbol{a_{d-k}}$.
Define the length of the $G$-cycle on the coefficient  $a_{d-k}$ to be this integer $r_k$. 

Observe that for any PP $P(x)$, there is an integer $r \ge 1$, such that the sequence of iterates of the $G$-map shown in (3)  forms a cycle. 
\begin{dfn} 
\label{def:G-cycle on P(x)}
The sequence of iterates of the $G$-map on the PP $P(x)$, namely 
\[P(x), G(P(x)),  G^2(P(x)), \dots,~ G^r(P(x)) = P(x),\] 
is called the {\bf $\boldsymbol{G}$-cycle on $\boldsymbol{P(x)}$}. 
\end{dfn}

Consider the values of $k$ for which the coefficient $a_{d-k}$ in $P(x)$ is nonzero.
For all such $k$, 
$(1 \le k \le d)$, 
let $r_{a_{d-k}}=\min\limits_{1 \leq j \leq m}\{j~|~ip^j \pmod {p^m-1}= i \}$, where $a_{d-k}=t^i$ for some $i$.
The length of the $G$-cycle on $P(x)$ is $r=\max\limits_{k}\{r_{a_{d-k}}\}$.
It should be noted that, for any element $a$ in $GF(q)$, the length of the $G$-cycle containing $a$ is a divisor of $m$.
The length of the $G$-cycle on $P(x)$ is simply the least common multiple of the lengths of the $G$-cycles on all non-zero coefficients of $P(x)$.

To illustrate $G$-cycles on coefficients, consider $GF(2^4)$. 
The elements of $GF(2^4)$ are partitioned into 6 disjoint equivalence classes ({\ie $G$-cycles on coefficients})
by the $G$-map, namely the equivalence classes [0], [1], [2], [4], [6], and [8]. 
To see this, observe that  
\begin{align*}
    [0] &=\{0^{2^i}\}=\{0\} \text{ for all } i                    & [6] &=[t^5]= \{(t^5)^{2^i}\} \text{ for all } i\\
    [1] &=[t^0]= \{(t^0)^{2^i}\}=\{1\} \text{ for all } i                    & &=\{(t^5)^{2^0},(t^5)^{2^1},(t^5)^{2^2},\dots\} \\
    [2] &=[t^1]= \{(t^1)^{2^i}\} \text{ for all } i                          & &=\{t^5,t^{10},t^5,\dots \}\\
        &=\{t^{2^0},t^{2^1},t^{2^2},t^{2^3},t^{2^4},\dots\}           & &=\{6,11\} \\
        &=\{t^1,t^2, t^4, t^8, t^1,\dots \}                       & [8] &= [t^7]=\{(t^7)^{2^i}\} \text{ for all } i\\
        &=\{2,3,5,9\}                                                 & &=\{(t^7)^{2^0},(t^7)^{2^1},(t^7)^{2^2},(t^7)^{2^3},(t^7)^{2^4}, \dots\}\\     
    [4] &= [t^3]= \{(t^3)^{2^i}\} \text{ for all } i                          & &=\{t^7,t^{14}, t^{13}, t^{11}, t^7,\dots \}\\
        &=\{(t^3)^{2^0},(t^3)^{2^1},(t^3)^{2^2},(t^3)^{2^3},(t^3)^{2^4}, \dots\}    & &=\{8,12,14,15\}\\ 
        &=\{t^3,t^6, t^{12}, t^9, t^3,\dots \}\\
        &=\{4,7,10,13\} 
\end{align*}

\noindent where $mod~{15}$ arithmetic is used in the exponents. 
The size of each equivalence class is the length of the corresponding $G$-cycle on the coefficient $a_{d-k}$. 
In this case, there are cycles of lengths 1,2, and 4.

To illustrate the computation of the $G$-cycle on $P(x)$, consider the degree 7 nPP $P(x) = x^7+x^5+8x^4+6x^2+4x$ over $GF(2^4)$. 
For all $k, (1 \le k \le d)$, the non-zero coefficients $a_{d-k}$ are $a_5=1, a_4=8, a_2=6$ and $a_1=4$, which correspond to $k=2,3,5$, and 6, respectively.
For example, the $G$-cycle for the nonzero coefficient $a_4=8=t^7$, is $[8]=[t^7]=\{8,12,14,15\}$, 
which gives the 4 successive coefficients of $x^4$ in the iterates of the $G$-map on $P(x)$ shown below.
We compute the length of the $G$-cycle for each nonzero coefficient $a_{d-k}$, namely, $r_{a_{d-k}}=\min\limits_{1 \leq j \leq m}\{j~|~ip^j \pmod {p^m-1}= i \}$ resulting in
$r_{a_5}=r_1=1$, $r_{a_4}=r_8=4$, $r_{a_2}=r_6=2$, $r_{a_1}=r_4=4$.   
By Definition \ref{def:G-cycle on P(x)}, the length of the $G$-cycle on $P(x)$ is $\max\limits_{k}\{r_{a_{d-k}}\}=4$,
as verified below.

\begin{flalign*}
    P(x) &= x^7+x^5+8x^4+6x^2+4x\\
    G(P(x)) &= 1^2x^7+1^2x^5+8^2x^4+6^2x^2+4^2x\\
            &= (t^0){^2}x^7+(t^0){^2}x^5+(t^7){^2}x^4+(t^5){^2}x^2+(t^3){^2}x\\
            &= (t^0)x^7+(t^0)x^5+(t^{14})x^4+(t^{10})x^2+(t^6)x\\
            &= x^7+x^5+15x^4+11x^2+7x\\
    G^2(P(x)) &= (t^0){^2}x^7+(t^0){^2}x^5+(t^{14}){^2}x^4+(t^{10}){^2}x^2+(t^6){^2}x\\
            &= (t^0)x^7+(t^0)x^5+(t^{28})x^4+(t^{20})x^2+(t^{12})x\\
            &= (t^0)x^7+(t^0)x^5+(t^{13})x^4+(t^{5})x^2+(t^{12})x\\
            &= x^7+x^5+14x^4+6x^2+13x\\
    G^3(P(x)) &= (t^0){^2}x^7+(t^0){^2}x^5+(t^{13}){^2}x^4+(t^{5}){^2}x^2+(t^{12}){^2}x\\
            &= (t^0)x^7+(t^0)x^5+(t^{26})x^4+(t^{10})x^2+(t^{24})x\\
            &= (t^0)x^7+(t^0)x^5+(t^{11})x^4+(t^{10})x^2+(t^{9})x\\
            &= x^7+x^5+12x^4+11x^2+10x\\
    G^4(P(x)) &= (t^0){^2}x^7+(t^0){^2}x^5+(t^{11}){^2}x^4+(t^{10}){^2}x^2+(t^{9}){^2}x\\
            &= (t^0)x^7+(t^0)x^5+(t^{22})x^4+(t^{20})x^2+(t^{18})x\\
            &= (t^0)x^7+(t^0)x^5+(t^{7})x^4+(t^{5})x^2+(t^{3})x\\
            &= x^7+x^5+8x^4+6x^2+4x \\
            &=P(x) 
\end{flalign*}

Note that the lengths of $G$-cycles on each of the nonzero coefficients $a_d-k, (1 \le k \le 7)$  are 1, 4, 2 and 4, respectively, and their least common multiple, namely 4, is the length of the $G$-cycle on $P(x)$.

\subsection{Iterating the $F$-map and the $G$-map}
\label{subsec:Fmap+Gmap}

We have introduced two functions, the $F$-map and $G$-map, that transform PPs into other PPs and nPPs into other nPPs. 
These functions can be applied sequentially. 
For example, we can represent the application of the $F$ and $G$ maps alternately two times on the PP $P(x)$ by the sequence $(G \circ F \circ G \circ F)(P(x))$, meaning one first applies the $F$-map, then the $G$-map, the $F$-map, and finally the $G$-map again. 
The $F$-map and the $G$-map are used together to create larger equivalence classes, allowing for a faster search for PPs.

It is interesting to note that two different sequences of compositions can represent the same transformation. 
In the following we show that we can replace any sequence of compositions by an equivalent sequence in which all of the $G$-maps are applied first, followed by some number of $F$-maps. The number of $F$-maps is related to the field characteristic $p$. 
Our result is illustrated by the following diagram, which indicates that $(G \circ F)(P(x))$ is the same as $(F^p \circ G)(P(x))$, 
for all PPs $P(x)$. 


\begin{center}
\begin{tikzcd}
P(x) \arrow{r}{G} \arrow{d}[swap]{F} & G(P(x)) \arrow{d}{F^p} \\   
F(P(x)) \arrow[swap]{r}{G} & (F^p\circ G)(P(x))
\end{tikzcd}
\end{center}


In Lemma \ref{lemma:FplusG}, we show that any sequence of $F$-maps and $G$-maps is equivalent to a sequence $F^i,G^j$, 
for some $i ~ (0 \le i \le r)$ and some $j ~(0 \le j \le s)$, 
where $r$ is the length of the $F$-cycle and $s$ is the length of the $G$-cycle.

\begin{lemma} \label{lemma:FplusG}
For any PP $P(x)$, $(G \circ F)(P(x))$ $=$ $(F^p \circ G)(P(x))$. 
\end{lemma}

\begin{proof}
Let $P(x) = \sum_{k=0}^d a_{d-k} x^{d-k}$. 
If $a_{d-k}=0$ then the coefficients of $x^{d-k}$ in $(G \circ F)(P(x))$ and $(F^p \circ G)(P(x))$ are both 0.
Suppose that $a_{d-k}\ne 0$. Then $a_{d-k}=t^j$, for some $j$, $0 \le j \le q-2$.
Then the $x^{d-k}$-th term in $F(P(x))$,  $G(P(x))$,  $(G \circ F)(P(x))$,  and $(F^p \circ G)(P(x))$, respectively, are:
\begin{align*}
F(P(x))&=\dots + t^{j+k} x^{d-k}+\dots \\
G(P(x))&=\dots + t^{pj} x^{d-k}+\dots\\
(G \circ F)(P(x))&=\dots + t^{(j+k)^p} x^{d-k}+\dots=\dots + t^{p(j+k)} x^{d-k}+\dots\\
(F^p \circ G)(P(x))&=\dots + t^{pj+pk}x^{d-k}+\dots =\dots + t^{p(j+k)}x^{d-k}+\dots
\end{align*}
Hence, the coefficients of $x^{d-k}$, $0\le k\le d$ in $(G \circ F)(P(x))$ and $(F^p \circ G)(P(x))$ are equal and the lemma follows.
\end{proof}
For example, let $P(x)$ be an nPP over $GF(2^5)$.
Consider $(G \circ F \circ G \circ F)(P(x))$. 
Since $p=2$, by Lemma  \ref{lemma:FplusG} $(G \circ F)(P(x))$ = $(F^2 \circ G)(P(x))$.
So, $(G \circ F \circ G \circ F)(P(x)) = (G \circ F \circ F^2 \circ G)(P(x))  = (G \circ F^3 \circ G)(P(x))$. 
Then by an iterative use of  Lemma \ref{lemma:FplusG}, we get $(G \circ F^3 \circ G)(P(x)) = (F^6 \circ G \circ G)(P(x)) = (F^6 \circ G^2)(P(x))$.



\section{Equivalence classes based on the $F$-map and the $G$-map}
\label{sec:F-G equivalence classes}

The equivalence relations $\R$ and $\R_E$ which we described in Section \ref{sec:Normalized-PPs} allow a more efficient search for PPs by limiting the search to nPPs.
More inclusive equivalence relations would optimize the search even further, by allowing the search to be restricted to representatives of equivalence classes.
In this section, we introduce new equivalence relations, based on the $F$-map and the $G$-map and iterations of the maps, that merge equivalence classes, therby compressing the search space considerably.
We begin with an equivalence relation induced by the $F$-map.

\begin{dfn} 
\label{def:Relation_R_F}
Let $P(x)$ and $Q(x)$ be PPs . If $P(x)$ can be converted into $Q(x)$ by some sequence consisting of normalization operations and $F$-map operations, then $P(x)$ and $Q(x)$ are related by $\R_F$.
\end{dfn}

We have seen that the iterates of the $F$-map form a cycle.
Moreover, if $P(x)$ is an nPP, then the $F$-cycle on $P(x)$ is a cycle of nPPs.
We will show shortly that  $\R_F$ is an equivalence relation that is identical to the $\R_E$ relation defined on nPPs (or PPs).
We can choose any nPP in an $\R_F$ equivalence class to be the representative for the class, but for convenience we usually designate the nPP with the smallest coefficient 
of a specific degree to be the representative. 
To illustrate, let $P(x)$ be an nPP in some equivalence class, let the specific degree be $d-k$ for some $k$, and consider the the coefficient of $x^{d-k}$, namely $a_{d-k}$.
Let the multiplicative inverse of $a_{d-k}$ be $t^{ik}$ for some $i$.
Suppose the length of the $F$-cycle on $a_{d-k}$ is $q-1$.
Then, sequence $a_{d-k},~t^k  a_{d-k}, ~t^{2k} a_{d-k},~ \dots~,t^{rk \mod(q-1)}a_{d-k}$ includes every nonzero element in $GF(q)$.
Specifically, one element in the $F$-cycle is $a_{d-k}t^{ik}=1$.
In other words, in the nPP $F^i(P(x))$,  
the coefficient of $x^{d-k}$ takes the value 1. 
So, in this case we choose $F^i(P(x))$ to be the representative of the equivalence class, making 
the search for nPPs more efficient since the coefficient of $x^{d-k}$ can be fixed to 1. 
That is, if there is a nPP whose $(d-k)^{th}$ coefficient is nonzero, there is also one whose $(d-k)^{th}$ coefficient is equal to 1. 
Notice that if the length of the $F$-cycle on $P(x)$ is less than $q-1$, then there may be a nPP with a nonzero $(d-k)^{th}$ coefficient, but not one with the $(d-k)^{th}$ coefficient equal to 1. 
In that case, if the cycle length is $(q-1)/j$, for some $j>1$, then one needs to search for nPPs with a $(d-k)^{th}$ coefficient equal to each of the values $1,2,3, \dots ,j-1$.
Note that if there is an nPP whose  $(d-k)^{th}$ coefficient is zero, it would be chosen as the representative for the equivalence class.

So again, in our search for all permutation polynomials over $GF(q)$ for a given $q$, we can reduce the search to the space of normalized PPs which are representatives of an $F$-map equivalence class. 
In other words, when the degree of the polynomial is not a multiple of the field characteristic, we restrict the search to polynomials $a_{d}x^{d} + a_{d-1}x^{d-1} + \dots + a_1x + a_0$ where $a_d=1$, $a_0=a_{d-1}=0$, and $a_i$ ranges over all cycle representatives, where the $i^{th}$ coefficient has the maximum length cycle.
When the degree is a multiple of the field characteristic, and the characteristic is not 2, we restrict the search to polynomials $a_{d}x^{d} + a_{d-1}x^{d-1} + \dots + a_1x + a_0$ where $a_d=1$, $a_0=0$, and one of the following cases (1) $a_{d-1}=0$, and (2) $a_{d-1} \neq 0$ and $a_{d-2} = 0$.  There is a similar statement when the field characteristic is 2 and the degree of the polynomial is even.

\begin{lemma}
The equivalence relations $\R_F$ and $\R_E$ are the same.
\end{lemma}

\begin{proof}
Suppose that $P(x)$ and $Q(x)$ are $\R_E$-related. 
That is, suppose that $Q(x)=aP(sx+b)+c$ for $a,b,c,s\in GF(q)$ such that $a\ne 0$ and $s\ne 0$.
We show that $P(x)$ and $Q(x)$ are $\R_F$-related. 
Since $s\ne 0$, we have $s=t^i$, for some $i ~(0 \le i \le q-2)$.
Consider $P'(x)=F^{q-i-1}(x)=t^{d(q-i-1)} P(x/t^{q-i-1})$, the last equality due to Equation (\ref{eqn:alt-F-i-map}).
Then $P'(x)=t^{d(q-i-1)} P(sx)$ since $t^i$ is the inverse of $t^{q-i-1}$.
Take $a'=at^{di}$ and $b'=b/s$. 
Note that $a'\ne 0$. 
Clearly, $P'$ and $P''(x)=a'P'(x+b')+c$ are $\R_F$-related.
Then
\begin{align*}
P''(x)&=a'(t^{d(q-i-1)} P(s(x+b'))+c\\
&=at^{di}(t^{d(q-i-1)} P(sx+sb')+c\\
&=aP(sx+b)+c\\
&= Q(x).
\end{align*}
Therefore $P(x)$ and $Q(x)$ are $\R_F$-related. 

Suppose that $P(x)$ and $Q(x)$ are $\R_F$-related. 
We show that $P(x)$ and $Q(x)$ are $\R_E$-related. 
Let $P_1(x)=P(x),P_2(x),\dots,P_k(x)=Q(x)$ be a sequence of permutation polynomials such that, for any $i$, 
$P_i(x)$ and $P_{i+1}(x)$ are $\R_F$-related.
Suppose $P_{i+1}(x)=F^j(P_i(x))$, for some $j~(0 \le j \le q-2)$. 
Then, by Equation (\ref{eqn:alt-F-i-map}) in Section \ref{subsec:Fmap}, $P_{i+1}(x)=t^{dj}P_i(x/t^j)=aP_i(sx+b)+c$ for $a=t^{dj}\ne 0,~s=t^{-j}=t^{q-j-1}\ne 0$, and $b=c=0$. 
That is, $P_i(x)$ and $P_{i+1}(x)$ are $\R_E$-related.
\end{proof}

Recall that one of our goals is to make the search for nPPs more efficient. 
Lemmas \ref{lemma:nPPwithF-k-is-relatively-prime} and \ref{lemma:nPPwithF-any-k} below, show how iterations of the $F$-map can be used to restrict the value of one additional coefficient and hence decrease the size of the search space. 

Lemma \ref{lemma:nPPwithF-k-is-relatively-prime} shows that under certain conditions, a fourth coefficient can be restricted to the values 0 and 1.

\begin{lemma} \label{lemma:nPPwithF-k-is-relatively-prime}
Let $s$ denote the position fixed for normalization, $i.e.$ $s=d-1$ for c-normalization, $s=d-1$ or $s=d-2$ for m-normalization, and $s=d-r-1$ or $s=d-r$, where $r=2^i-1$ for some $i$, for b-normalization.
If there is a $k \ne s$ relatively prime to $q-1$,
then each equivalence class of $\R_F$ contains an nPP 
$Q(x) = b_dx^d + b_{d-1}x^{d-1} + \dots + b_{d-k}x^{d-k} + \dots + b_1x + b_0$
satisfying 
$b_d=1,~ b_{s}= b_0=0$ and $b_{d-k}\in\{0,1\}$, for some $k\in {1,\dots, d-1}, k \neq s$.
\end{lemma}

\begin{proof}
Let $P(x) = \sum_{k=0}^d a_{d-k}x^{d-k}$ be an arbitrary nPP from any equivalence class of $\R_F$. 
Then $a_d=1 ,a_{s}=0$, and $a_0=0$. 
On one hand, if $a_{d-k}$ is 0, then $Q(x)=P(x)$, and the lemma is proved. 

On the other hand, suppose that $a_{d-k}=t^i, i\in\{1,2,\dots,q-2\}$.
Since $k$ is relatively prime to $q-1$, there is a positive integer $j$ such that $i+jk=0\pmod{q-1}$.  
Let $Q$ be the polynomial obtained by applying the $F$ map $j$ times to $P$.
That is, let $Q(x)=\sum_{k=0}^d a_{d-k} t^{jk} x^{d-k}$
Then the coefficient of the $x^{d-k}$ term of $Q$ is $b_{d-k}=a_{d-k}t^{kj}=t^i t^{kj}=t^{i+kj}=1$.
Furthermore, $b_d=a_d=1$, $b_0=a_0=0$, and $b_{s}=a_st^{(d-s)j}=0t^{(d-s)j}=0$.
Hence Q is an nPP with the desired coefficients.
\end{proof}

It could be that there is no $k$ in the range $(2\le k\le d-2$) that is relatively prime to $q-1$ and Lemma \ref{lemma:nPPwithF-k-is-relatively-prime} cannot be applied. For example, when $q=31$ and $d=8$, then none of the integers in $\{2,3,4,5,6\}$ are relatively prime to $q-1$.
We now show that when this is the case, we can restrict the range of the coefficient of either $x^{d-1}$ or $x^{d-2}$ to at most 3 values.

\begin{lemma}
\label{lemma:nPPwithF-any-k}
Let $s$ denote the position fixed for normalization, $i.e.$ $s=d-1$ for c-normalization, $s=d-1$ or $s=d-2$ for m-normalization, and $s=d-r-1$ or $s=d-r$, where $r=2^i-1$ for some $i$, for b-normalization.
Then each equivalence class of $\R_F$ contains an nPP 
$Q(x) = b_dx^d + b_{d-1}x^{d-1} + \dots + b_1x + b_0$
satisfying 
$b_d=1,~ b_{s}= b_0=0$ and if $s = d-1$, then $b_{d-2}\in\{0,1,2\}$, otherwise, if $ s \geq d-2$, then  $b_{d-1}\in\{0,1\}$.
\end{lemma}

\begin{proof}
Let $P(x) = a_dx^d + a_{d-1}x^{d-1} + \dots + a_1x + a_0=\sum_{k=0}^d a_{d-k}x^{d-k}$ be an arbitrary nPP from any equivalence class of $\R_F$. 
Then $a_d=1 ,a_s=0$, and $a_0=0$. 
We will obtain a new nPP $Q(x)$ from $P(x)$ by using the $F$-map iteratively. 
The coefficients of $Q(x)$ are changed from the given $a_i$ to a new value $b_i$, for all i. 
Since the $F$-map does not change the first coefficient, $b_d=a_d=1$. 
Since the $F$-map does not change a zero coefficient, 
$b_s=a_s=0$, and $b_0=a_0=0$. 
It remains to prove that if $s=d-1$ then $b_{d-2}\in\{0,1,2\}$, otherwise, if $s \geq d-2$, then $b_{d-1}\in\{0,1\}$.

\begin{enumerate} [{Case} 1{:}]
\item $s = d-1$.\\ 
    If $a_{d-2}=0$, then the lemma is proved, as we can take $Q(x)=P(x)$. \\
    So, suppose $a_{d-2} \neq 0$.
    If 2 is relatively prime to $q-1$, then 
    Lemma \ref{lemma:nPPwithF-k-is-relatively-prime}  fixes $b_{d-2}$ to a value in $\{0,1\}$.
    So suppose 2 is not relatively prime to $q-1$.
    Then $q-1$ must be even, that is $gcd(2,q-1)=2$.
    Hence, the length of $F$-cycle on $a_{d-2}$ is $(q-1)/2$. 
    Substituting k=2 into the definition of the $F$-cycle on $a_{d-2}$ yields the sequence \[ a_{d-2},~t^2  a_{d-2}, ~t^{4} a_{d-2},~ \dots~,(t^{2r \mod(q-1)}a_{d-2}=a_{d-2}) \] where $r=(q-1)/2$.
    So there are two possibilities for the values  in $F$-cycle on $a_{d-2}$:
    \begin{itemize}
        \item the $F$-cycle contains $(q-1)/2$ odd values including 1, in which case $b_{d-2}$ can be fixed at 1,
        \item the $F$-cycle contains $(q-1)/2$ even values including 2, in which case $b_{d-2}$ can be fixed at 2.
    \end{itemize}
\item $s \geq d-2$.\\ 
    If $a_{d-1}=0$, then the lemma is proved, as we can take $Q(x)=P(x)$. \\  So, suppose $a_{d-1} \neq 0$.
    Since 1 and $q-1$ are relatively prime,  
    Lemma \ref{lemma:nPPwithF-k-is-relatively-prime}  fixes $b_{d-1}$ to a value in $\{0,1\}$.
\end{enumerate} 
Hence $Q(x)$ is an nPP with the desired coefficients.
\end{proof}

If Lemma \ref{lemma:nPPwithF-k-is-relatively-prime} does not apply, then the $F$-cycle is not of length $q-1$. Since the length of the $F$-cycle must divide $q-1$, it follows from Lemma \ref{lemma:nPPwithF-any-k} that, when the coefficient $a_{d-2}$ is not zero, the $F$-cycle is of length $(q-1)/2$. That is, as the coefficient $a_{d-2}$ is multiplied by $t^2$ with each iteration of the $F$-map, to return in the cycle to the same coefficient $a_{d-2}$, the length of the cycle must be $r$, where $r$ is the smallest integer such that $(t^2)^r = 1$. It follows that $r = (q-1)/2$, as this is the  smallest $r$ such that $2r = 0$ $mod$ $(q-1)$.

By Lemmas \ref{lemma:nPPwithF-k-is-relatively-prime} and \ref{lemma:nPPwithF-any-k} we have shown that every equivalence class of $\R_F$ contains an nPP in which four of its coefficients are either fixed or are limited to at most three choices. 
We have implemented this in our search algorithm, {\tt iBlast} (described in Section \ref{subsec:iBlast}).
{\tt iBlast} searches for equivalence class representatives in O($dq^{d-2}$) time rather than O($dq^{d-1}$) time, as was previously described in Section \ref{sec:intro}.


In Section \ref{subsec:Fmap+Gmap} we discussed iterations of the $F$-map and the $G$-map. 
We now introduce a new equivalence relation, $\R_{F,G}$, based on such iterations. 

\begin{dfn} 
\label{def:Relation_R_FG}
Let P(x) and Q(x) be nPPs of degree d, for some d. If P(x) can be converted into Q(x) by some sequence consisting of $F$-map and $G$-map operations, then P(x) and Q(x) are related by $\R_{F,G}$.
\end{dfn}


We present a theorem that characterizes the nPPs in the equivalence classes of $\R_{F,G}$, and gives the number of nPPs in each equivalence class. 

\begin{theorem} 
\label{th:sizeof-R_FG-equivclasses}
Let $P(x)$
be a nPP.  
Then the equivalence class of $\R_{F,G}$ containing $P(x)$ is \[[P(x)]=\{F^i(G^j(P))~|~0\le i\le r-1, 0\le j\le s-1\},\]
where $r$ and $s$ are the lengths of the $F$- and $G$-cycles, respectively, and the number of nPPs in $[P(x)]$ is $rs$.
\end{theorem}

\begin{proof}
This follows from Lemma \ref{lemma:FplusG}.
Clearly, $[P(x)]$ is a subset of the equivalence class of $\R_{F,G}$ containing $P(x)$. 
Let $Q(x)$ be a polynomial of $[P(x)]$, that is, let $Q(x) = (F^i \circ G^j)(P)$, for some $i,j$.
Observe that $F(Q(x))=(F \circ F^i \circ G^j)(P(x))=(F^{i+1}\circ G^j)(P(x))$, so $F(Q(x))\in [P(x)]$. 
It remains to show that $G(Q(x))\in [P(x)]$.
If $i=0$ then $G(Q(x))=G^{j+1}(P(x))$, so $G(Q(x))\in [P(x)]$.
Suppose $i\ge 1$.
Then $G(Q(x))=(G \circ F^i \circ G^j)(P(x))$.
As the iterates of the $F$-map and the $G$-map form cycles of lengths $r$ and $s$, respectively, the size of the exponents of $F$ and $G$ are computed by arithmetic $\pmod r \text{ and}\pmod s$, respectively. 
Applying Lemma \ref{lemma:FplusG} $i$ times, we have 
\begin{align*}
    G(Q(x)) &= (G \circ F^i \circ  G^j )(P(x))\\
    &= (F^p \circ G \circ F^{i-1} \circ  G^j )(P(x))\\
    &= (F^{p+1} \circ G \circ F^{i-2} \circ  G^j )(P(x))\\
    &= \dots\\
    & = (F^{pi} \circ G^{j+1})(P(x))\\
    &\in [P(x)]
\end{align*}
By algebra, the number of nPPs in $[P(x)]$ is $rs$.
The theorem follows. 
\end{proof}

For example, consider the nPP $P(x)$ = $x^8+4x^4+16x^2+3x$ for $GF(2^5)$, with the primitive polynomial $x^5+x^3+x^2+x+1$. 
As there are 31 non-zero elements in $GF(31)$, and 31 is prime, the length of the $F$-cycle on $P(x)$ is 31.  
The iterates of $G$-map on the coefficient 16
give the $G$-cycle $\{16,31,30,27,24\}$, which has length 5.
So, $[P(x)]$ denotes an equivalence class with $31*5=155$ nPPs.

Theorem \ref{th:sizeof-R_FG-equivclasses} shows that combined iterations of the $F$-map and the $G$-map on an nPP $P(x)$ can condense smaller equivalence classes induced by the individual maps into fewer (and larger) equivalence classes. 
So although the $G$-map by itself (unlike the $F$-map) does not allow an additional coefficient to be fixed, when used with the $F$-map, it does allow the search to include fewer equivalence classes. 
This increases the efficiency of the search for PPs, although not by an order of magnitude.


\section{Algorithms for computing equivalence classes of nPPs}
\label{sec:algorithms}

\subsection{Algorithm {\tt iBlast}}
\label{subsec:iBlast}

We have implemented an  algorithm called {\tt iBlast} (whose name is derived from our permuted first name initials) that uses Lemmas \ref{lemma:nPPwithF-k-is-relatively-prime} and \ref{lemma:nPPwithF-any-k}, and Theorem \ref{th:sizeof-R_FG-equivclasses} to make the enumeration of  nPPs more efficient. 
As previously remarked, {\tt iBlast} searches for equivalence class representatives in O($dq^{d-2}$) time rather than O($dq^{d-1}$) time (the latter being the time bound of previous algorithms), so we are able to compute new results for various $q$ and $d$. 
We exhibit many of these in Table \ref{tbl-deg11} and give specific details in Tables \ref{tb1-Num-nPPs} through \ref{tbl-EqCl-many}.

\begin{algorithm}[!htb]
\label{algorithm:iBlast}
\DontPrintSemicolon
\SetKwRepeat{Do}{do}{while}
\textbf{if} $p \nmid d$ \textbf{then} \textup{Create the set $\cal M$ of all possible masks corresponding to c-normalization} \;
\textbf{else if} $p = 2$ \textbf{then} \textup{Create the set $\cal M$ of all possible masks corresponding to b-normalization} \;
\textbf{else} \textup{Create the set $\cal M$ of all possible masks corresponding to m-normalization} \;
$\cal S = \emptyset$ \;
\ForEach{ \textup{mask} $m \in \cal M$ }
{
  \textup{Use the $F$-map and $G$-map functions to select the optimal coefficient in $m$ to additionally fix}\;
  \textup{$currentPolynomial$ = its default value where each fixed coefficient is assigned its designated value, and each unfixed coefficient is assigned 1}\;
  \Do{currentPolynomial \textup{is not maximized}}
  {
    \If{$currentPolynomial \textup{ is a PP \textbf{and} } currentPolynomial \notin \cal S$}
    {
      ${\cal T} =$ \textup{the set of all PPs that are $\R_{F,G}$-related to $currentPolynomial$}\;
      \If{$p \mid d$}
      {
        \ForEach{\textup{PP} $P \in \cal T$ \textup{\textbf{and}} $b \in GF(q)$}
        {
            ${\cal T} = {\cal T} \cup P(x+b)$
        }
      }
      ${\cal S} = {\cal S} \cup {\cal T}$\;
    }
    increment $currentPolynomial$\;
  }
}
\Return{$\cal S$}
\caption{iBlast}
\end{algorithm}

{\tt iBlast} computes the set ${\cal S}$ of all nPPs in $GF(q)$ of degree $d$.
Define a $mask$ as an array with values in $\{0,1\}$ associated with a polynomial's coefficients. 
0 designates that the corresponding  coefficient should be fixed at a particular value, and 1 designates the coefficient is both unfixed and nonzero. 
The purpose of using masks is to maximize our usage of the $F$-map and $G$-map functions. 
Since both the $F$-map and $G$-map  take the symbol 0 to itself, masks allow us to consider only nonzero elements when selecting an additional coefficient to fix. 

For example, consider the search for nPPs of degree 5 over $GF(25)$. 
Since $p \mid d$ and $p > 2$, the masks must correspond to m-normalization, that is, to PPs in the form $a_{5}x^5+a_{4}x^4+a_{3}x^3+a_{2}x^2+a_{1}x+a_{0}$ where $a_{5}=1,~ a_{0}=0$, and either $a_{4}=0$ or $a_{3}=0$.
Since $a_{5}$ is always fixed at 1, and $a_{0}$ is always fixed at 0, we can first consider all cases where $a_{4}$ is fixed at 0. 
{\tt iBlast} generates $2^3=8$ masks in the form $[a_{5},a_{4},a_{3},a_{2},a_{1},a_{0}]$ to account for all combinations where $a_{3}, a_{2},$ and $a_{1}$ are either fixed at 0 or are unfixed and nonzero.
The eight masks are $[0,0,0,0,0,0]$, $[0,0,0,0,1,0]$, $[0,0,0,1,0,0]$, $[0,0,0,1,1,0]$, $[0,0,1,0,0,0]$, $[0,0,1,0,1,0]$, $[0,0,1,1,0,0]$, $[0,0,1,1,1,0]$.
{\tt iBlast} also creates an additional 4 masks for the cases where $a_{3}$ is fixed at 0, and $a_{2}$ is unfixed and nonzero, namely $[0,1,0,0,0,0]$, $[0,1,0,0,1,0]$, $[0,1,0,1,0,0]$, $[0,1,0,1,1,0]$.

{\tt iBlast} then iterates through the masks to search for nPPs. 
Each mask has at least 3 coefficients fixed due to normalization, and
Lemmas \ref{lemma:nPPwithF-k-is-relatively-prime} and \ref{lemma:nPPwithF-any-k}
fix an additional coefficient, giving a total of 4 fixed coefficients.

\begin{table}[htb]
\centering
\vspace*{4mm}
\begin{tabular}{|c|c|c|c|}
\hline
\multicolumn{2}{|c|}{\bfseries  $G$-cycles} & {\bfseries $F$-cycles for $a_{4}$} & {\bfseries $F$-cycles for $a_{2}$} \\
\hline
$\{0\}$ & $\{9, 17\}$& $\{0\}$& $\{0\}$\\
$\{1\}$& $\{10, 22\}$& $\{1, 2, 3, \dots, 24\}$& $\{1, 4, 7, 10, 13, 16, 19, 22\}$\\
$\{2, 6\}$& $\{13\}$& & $\{2, 5, 8, 11, 14, 17, 20, 23\}$\\
$\{3, 11\}$& $\{14, 18\}$& & $\{3, 6, 9, 12, 15, 18, 21, 24\}$\\
$\{4, 16\}$& $\{15, 23\}$& & \\
$\{5, 21\}$& $\{19\}$& & \\
$\{7\}$& $\{20, 24\}$& & \\
$\{8, 12\}$& & & \\ 
\hline
\end{tabular}
\caption{$G$-cycles and $F$-cycles for the mask $[0,1,0,1,0,0]$ in $GF(25)$ degree 5.}
\label{tbl-fgcycles-25-5}
\end{table}

For example, consider the $G$-cycles and $F$-cycles for the mask $[0,1,0,1,0,0]$ in Table \ref{tbl-fgcycles-25-5}.
Observe that for the $F$-cycles of $a_{2}$ we can use the $G$-map to map every element in the $F$-cycle that contains 2 (\ie{the $F$-cycle $\{2, 5, 8, 11, 14, 17, 20, 23\}$}) to some element in the $F$-cycle that contains 6 (\ie{the $F$-cycle $\{3, 6, 9, 12, 15, 18, 21, 24\}$}). 
If we were to therefore consider fixing the coefficient $a_{2}$, we would only need to check polynomials where $a_{2}$ had a value of 1 or 3. However, $a_{4}$ is the better choice as it can be fixed to the single value 1 because for $a_{4}$, the $F$-cycle that contains 1 also contains all non-zero elements of $GF(25)$.

{\tt iBlast} then begins the search by creating a polynomial variable, $currentPolynomial$, and setting its coefficients to their assigned value if they are fixed, or 1 if they are unfixed. 
For our example mask $[0,1,0,1,0,0]$, this would correspond to the polynomial $x^5+x^4+x^2$.
{\tt iBlast} increments the polynomial as a $d+1$-tuple and adds 1 to the lowest degree, unfixed coefficient. 
If a coefficient exceeds its maximum  value, {\tt iBlast} resets it to 1, since unfixed coefficients are nonzero, and the next highest unfixed coefficient is incremented.
Note that if the coefficient fixed by the $F$-map and $G$-map does not reduce to a single value, we can include it in our increment function, but it only increments through the minimum necessary values as determined by the $F$-cycles and $G$-cycles.

Each time $currentPolynomial$ is incremented, it is checked to determine if it is a PP. 
If $currentPolynomial$ is a new PP, \ie{it is not already in the set $\S$,} {\tt iBlast} uses the $F$-map and $G$-map functions to create the set $\cal T$ of all PPs that are $\R_{F,G}$-related to $currentPolynomial$.

If $p \mid d$ and $p > 2$, m-normalization applies.
In this case, for each PP in $\cal T$,  {\tt iBlast} computes $P(x+b)$, for all $b$ in $GF(q)$, and puts the resulting PPs in the set $\cal T$. 
Note that this must be done for any search where $p \mid d$ since both m-normalization and b-normalization make use of this normalization operation.
Finally {\tt iBlast} adds all of the newly generated PPs to S by setting $\S = \S \cup \cal T$ and then continues to increment $currentPolynomial$ until it reaches its maximum configuration.

Reaching the maximum configuration signifies that the process can be repeated for the next mask.
{\tt iBlast} terminates when all of the masks have been considered. 
At this point, the set $\S$ contains all nPPs in $GF(q)$ of degree $d$.

\subsection{Algorithm 2}


Given the set $\cal S$ of all nPPs in $GF(q)$, the following algorithm will create a set of equivalence class representatives.

While $\cal S$ is not empty, select an arbitrary $P(x)$ in $\cal S$. Compute the set $\cal T$ of all nPPs that are $\R_{F,G}$-related to $P(x)$. If $p \mid d$, additionally compute all possible combinations of $P(x+b)$ for each nPP in $\cal T$, adding each result to $\cal T$. Designate $P(x)$ as an equivalence class representative, and ${\cal S} = {\cal S} - {\cal T}$.

\section{Results}
\label{sec:results}

As stated in Section \ref{sec:intro}, a brute force search for degree $d$ permutation polynomials over $GF(q)$ 
would require $O(dq^{d+2})$ time. 
Normalization operations defined in \cite{lidl88}, for PPs in which $p \nmid d$, fixes three of the coefficients and therefore requires $O(dq^{d-1})$ time. 
We refer to this type of normalization as  c-normalization.
In this paper, we have succeeded in improving the time bound by an order of magnitude, that is, lowering it to $O(dq^{d-2})$, by fixing an additional coefficient. 
This improvement applies to a larger class of PPs, namely all c-normalized PPs, all PPs for which  $p \nmid d$ and $p>2$ (\ie, m-normalized PPs), and all PPs for which $p \nmid d$ and $p>2$, except when $d=2^i-2$ for some $i \geq 2$ (\ie, b-normalized PPs).
We have done this by expanding the definition of normalization to include m-normalization  and b-normalization, and by introducing four new equivalence relations on PPs and nPPs, namely, $\R_E,R_F, R_G, \text{ and } R_{F,G}$.
We have been able to reduce the search space for PPs by limiting the search to equivalence class representatives.
In addition, equivalence classes allow a more succinct categorization of PPs, since the equivalence classes can include quite a number of PPs.
Furthermore, our new techniques apply to arbitrary $q$ and $d$.

We implemented our search for equivalence classes in the algorithm {\tt iBlast} and computed many new results which are shown in Table \ref{tb1-Num-nPPs}.
For almost all $q\le 100$ and $d\le 10$, {\tt iBlast} found all nPPs and all equivalence classes.
Table \ref{tb1-Num-nPPs} lists the number of nPPs, the number of  equivalence classes, and total number of PPs for $q \le 97$ and degree $d$, where $6 \le d \le 10$.
Note that Table \ref{tb1-Num-nPPs} has columns for $d \in \{6,7,8,9,10\}$. 
For degrees $d \le 5$, all PPs have been described, for example in \cite{chu2004}. More recent work \cite{Fan19a,Fan19,Fan19b,li10,Shallue-Wanless-pp-13} gives all PPs of degree $d \leq 8$; however, we list our results in Table \ref{tb1-Num-nPPs} for completeness.
Table \ref{tb1-Num-nPPs} does not have columns for $d \ge 11$, because the computations become too time consuming (at least for large $q$).
However, we have been able to compute all degree 11 PPs over $GF(q)$, for powers of primes $q$ ($16 \le q \le 32$). 
The results are listed in Table \ref{tbl-deg11}.
The sum of $N_d(q)$, for $1 \leq  d \leq 11$, for prime powers $q$, also gives improved lower bounds for $M(q,q-11)$ shown in Table \ref{tbl-M(n,d)}. For example, we show that M(16,5) $\ge 5,112,053,760$, which improves on the lower bound given in \cite{wang17}. We also computed $N_{12}(17) = 
68,126,982,656, N_{12}(19) = 46,631,675,376, N_{12}(23) = 13,755,394,444 $. This yields $M(17,5) \ge 72,377,516,320$$, M(19,7) \ge  $, and $M(23,11) \ge 14,341,972,920$.

\begin{table}[htb]
\centering
\vspace*{4mm}
\begin{tabular}{|r|r|r|r|}
\hline
\bfseries  $q$ & $N_{11}(q)$ & number of nPPs & number of equivalence classes \\
\hline
16 & 4,751,093,760 & 1,237,264 & 20,663 \\
17 & 4,001,494,000 & 865,375 & 54,225 \\
19 & 2,431,915,488 & 374,256 & 20,874\\
23 & 0 & 0 & 0 \\
25 & 6,509,295,000 & 433,953 & 9,266 \\
27 & 2,826,989,100& 149,150& 2,060\\
29 & 1,014,518,484 & 43,083 & 1,639 \\
31 & 385,053,480 & 13,356 & 507 \\ 
32 & 190,940,160 & 6,015 & 51 \\
\hline
\end{tabular}
\caption{Number of PPs and nPPs for degree 11 polynomials over $GF(q)$.}
\label{tbl-deg11}
\end{table}

We also provide a website that explicitly lists nPPs and PPs of degree $1\le d\le 10$ and the values of $n\le 100$ at \url{https://personal.utdallas.edu/~bdm170430/npps/}.
The nPPs are computed using normalizations from Table \ref{tbl-NormTypes}.  
The equivalence classes are defined using the relations $\R$, $\R_F$, and $\R_{F,G}$. 
The number of equivalence classes is given in parentheses. The total number of PPs is computed from the number of nPPs. 
In looking at Tables \ref{tb1-Num-nPPs}, and Tables \ref{tbl-EqCl-11-7} through \ref{tbl-EqCl-many}, one should keep in mind that the specific nPPs listed in the tables are for the stated primitive polynomial and for our naming convention for elements of $GF(q)$. Our results for degrees 7 and 8 agree with those listed in \cite{Fan19}, \cite{Fan19b}, and \cite{Fan19a}, except for differences caused by naming conventions.

In some cases, there is a small number of equivalence classes representing a large number of PPs. 
For example, 
$N_8(27)= 6,899,256$, and there are 364 nPPs, but only 6 equivalence classes as shown in Table \ref{tbl-EqCl-27-8}.  Compare our list of classes of degree 8 PPs over GF(27) with the list given in \cite{Fan19b}, which has 26 nPPs. The difference is that we often combined three of his nPPs into one class by the use of our $R_{F,G}$ relation.
This gives a method to make rather concise representations of large sets of PPs. 
Again, for example, $N_9(32)= 9,872,384$, and there are 311 nPPs, but  only 7 equivalence classes as shown in Table \ref{tbl-EqCl-32-9}. 
Observe that, for the second through sixth nPPs in Table \ref{tbl-EqCl-32-9}, there is a term, either $x^3$, $x^6$, or $x^7$, with coefficient 1. 
Since the number of non-zero elements is 31, which is prime, the length  of the $F$-map has length 31. 
Since $G(1)=1$, the $G$-map produces no other nPPs in the  equivalence classes. 
So, each equivalence class has 31 elements, as defined by the $F$-map alone. 
For the last PP in Table \ref{tbl-EqCl-32-9}, again the length  of the $F$-cycle on that PP is 31, but now the length  of the $G$-cycle is applied, for example, to the coefficient 16, has length 5, namely (16,31,30,27,24). 
Thus, the equivalence class of nPPs produced by the $F$-map and $G$-map combined is 31 times 5 or 155.  

We note that it was stated in \cite{Sobhani2019} that 32 hours were required to compute the RS code for $q=32$ and $d=5$; 
for $q=32$ and $d=7$, the authors
were able to compute only the size of the set of PPs, not the set itself.
Our algorithm {\tt iBlast} allows for the computation of all nPPs, for $q = 32$ and $d = 7$, in a few seconds and, for $q = 32$ and $d = 8$, in about 10 minutes. 

It should be noted that in many of our tables (Tables \ref{tbl-EqCl-11-7} through \ref{tbl-EqCl-many}) we use variables to further reduce the table's size. 
For example, for $q = 25$ and $d = 9$ (see Table \ref{tbl-EqCl-25-9}) there are 38 equivalence classes, which is rather large. 
Instead, in Table \ref{tbl-EqCl-25-9}, we use a variable $a$, with specified values, which allows the table to have 21 classes rather than 38. 
For example, the first class shown is $x^9+ax$, where $a \in \{1,7\}$. This is the union of two equivalence classes, namely those with the representatives $x^9+x$ and $x^9+7x$.
The variable $a$ is also used in other tables.
In Table \ref{tbl-EqCl-25-10}, $\alpha(x)$ represents three sets of values for low order coefficients in the polynomial shown in the first column.

%
%
Our results on the number of PPs also give several new lower bounds for $M(n,D)$. These are given in Table \ref{tbl-M(n,d)}. 
Several additional improved lower bounds can be obtained from those shown in Table \ref{tbl-M(n,d)} using the well known inequality $M(n,D-1) \ge M(n,D)/n$ \cite {chu2004} or the operation of contraction \cite {bls-18}.

Lidl and Mullen \cite{lidl88,lidl93} give a list of several interesting open questions about PPs. 
For example, they ask about $complete$ PPs, where a PP $P(x)$ is called $complete$ if $P(x)+x$ is also a PP. 
There are several complete PPs in our tables. For example, $x^9$ and $x^9+x$ are both PPs for $q=3^4$.

For the reader's convenience, we provide a master list of tables in Table \ref{tbl-ListOfTables}.

\begin{table}[htb]
\centering
\vspace*{4mm}
\begin{tabular}{|l|l|}
\hline
Table \ref{tbl-deg11}  	 & 	Number of PPs and nPPs for degree 11 polynomials over $GF(q)$	\\
Table \ref{tb1-Num-nPPs}	 & 	Number of nPPs, Equivalence Classes, and Total number of PPs for $q \le 97$ and degree $d$,\\
    & where $6 \le d \le 10$	\\
Table \ref{tbl-M(n,d)}	 & 	New lower bounds for $M(n,D)$	\\
Table \ref{tbl-EqCl-11-7}	 & 	nPP Equivalence Classes for $q=11$, degree 7, and primitive polynomial $x+4$	\\
Table \ref{tbl-EqCl-23-7}	 & 	nPP Equivalence Classes for $q=23$, degree 7, and primitive polynomial $x+12$	\\
Table \ref{tbl-EqCl-25-7}	 & 	nPP Equivalence Classes for $q=25$, degree 7, and primitive polynomial $x^2+3x+3$	\\
Table \ref{tbl-EqCl-25-9}	 & 	nPP Equivalence Classes for $q=25$, degree 9, and primitive polynomial $x^2+3x+3$	\\
Table \ref{tbl-EqCl-25-10}	 & 	nPP Equivalence Classes for $q=25$, degree 10, and primitive polynomial $x^2+3x+3$	\\
Table \ref{tbl-EqCl-27-8}	 & 	nPP Equivalence Classes for $q=27$, degree 8, and primitive polynomial $x^3+2x+1$	\\
Table \ref{tbl-EqCl-27-9}	 & 	nPP Equivalence Classes for $q=27$, degree 9, and primitive polynomial $x^3+2x+1$	\\
Table \ref{tbl-EqCl-32-8}	 & 	nPP Equivalence Classes for $q=32$, degree 8,\\
    & and primitive polynomial $x^5+x^3+x^2+x+1$	\\
Table \ref{tbl-EqCl-32-9}	 & 	nPP Equivalence Classes for $q=32$, degree 9,\\
    & and primitive polynomial $x^5+x^3+x^2+x+1$	\\
Table \ref{tbl-EqCl-37-7}	 & 	nPP Equivalence Classes for $q=37$, degree 7, and primitive polynomial $x+13$	\\
Table \ref{tbl-EqCl-41-9}	 & 	nPP Equivalence Classes for $q=41$, degree 9, and primitive polynomial $x+17$	\\
Table \ref{tbl-EqCl-43-9}	 & 	nPP Equivalence Classes for $q=43$, degree 9, and primitive polynomial $x+13$	\\
Table \ref{tbl-EqCl-43-10}	 & 	nPP Equivalence Classes for $q=43$, degree 10, and primitive polynomial $x+13$	\\
Table \ref{tb-EqCl-47-9}	 & 	nPP Equivalence Classes for $q=47$, degree 9, and primitive polynomial $x+12$	\\
Table \ref{tbl-EqCl-49-9}	 & 	nPP Equivalence Classes for $q=49$, degree 9, and primitive polynomial $x^2+6x+3$	\\
Table \ref{tbl-EqCl-49-10}	 & 	nPP Equivalence Classes for $q=49$, degree 10, and primitive polynomial $x^2+6x+3$	\\
Table \ref{tbl-EqCl-71-9}	 & 	nPP Equivalence Classes for $q=71$, degree 9, and primitive polynomial $x+38$	\\
Table \ref{tbl-EqCl-73-7}	 & 	nPP Equivalence Classes for $q=73$, degree 7, and primitive polynomial $x+20$	\\
Table \ref{tbl-EqCl-79-7}	 & 	nPP Equivalence Classes for $q=79$, degree 7, and primitive polynomial $x+25$	\\
Table \ref{tbl-EqCl-81-9}	 & 	nPP Equivalence Classes for $q=81$, degree 9, and primitive polynomial $x^4+2x^3+2$	\\
Table \ref{tbl-EqCl-89-7}	 & 	nPP Equivalence Classes for $q=89$, degree 7, and primitive polynomial $x+76$	\\
Table \ref{tbl-EqCl-many}	 & 	Number of nPPs in Equivalence Classes for prime powers q and degree d,\\
    & for given primitive polynomial	\\
\hline
\end{tabular}
\caption{List of tables of results.}
\label{tbl-ListOfTables}
\end{table}

\subsection*{Acknowledgments} 

We thank Professor Carlos Arreche of the Mathematics Department at UTD for bringing our attention to Galois field orbits.
We also want to thank Professor Xiang Fan of Mathematics Department Sun Yat Sen University in Guangzhou, China for his valuable comments and suggestions, including the use of Lucas's theorem.
 
\bibliographystyle{abbrv} 

\pagebreak
\begin{longtable}{|l|lllll|}
\caption{Number of nPPs, Equivalence Classes, and Total number of PPs for $q \le 97$ and degree $d$, where $6 \le d \le 10$} \\
\hline \multicolumn{1}{|l|}{\textbf{q}} & \multicolumn{1}{l}{\textbf{6}} & \multicolumn{1}{l}{\textbf{7}} & \multicolumn{1}{l}{\textbf{8}} & \multicolumn{1}{l}{\textbf{9}} & \multicolumn{1}{l|}{\textbf{10}} \\ \hline
\endfirsthead
\multicolumn{6}{c}%
{{\bfseries \tablename\ \thetable{} -- continued from previous page}} \\
\hline \multicolumn{1}{|l|}{\textbf{q}} & \multicolumn{1}{l}{\textbf{6}} & \multicolumn{1}{l}{\textbf{7}} & \multicolumn{1}{l}{\textbf{8}} & \multicolumn{1}{l}{\textbf{9}} & \multicolumn{1}{l|}{\textbf{10}} \\ \hline
\endhead
\hline \multicolumn{6}{|r|}{{Continued on next page}} \\ \hline
\endfoot
\endlastfoot
\textbf{11} & & & & & \\*
\hspace{2mm}nPPs (ECs) & 24 (4) & 225 (28) & 2,754 (277) & 29,985 (3,036) & \\*
\rule[-1.5ex]{0pt}{0pt}\hspace{2mm}Total PPs & 29,040 & 272,250 & 3,332,340 & 36,281,850 & \\
\textbf{13} & & & & & \\*
\hspace{2mm}nPPs (ECs) & 0 & 115 (15) & 1,380 (117) & 16,740 (1,422) & 218,020 (18,193) \\*
\rule[-1.5ex]{0pt}{0pt}\hspace{2mm}Total PPs & 0 & 233,220 & 2,798,640 & 33,948,720 & 442,144,560 \\
\textbf{16} & & & & & \\*
\hspace{2mm}nPPs (ECs) & 840 (3) & 216 (7) & 14,816 (57) & 4,200 (74) & 1,417,600 (1,786) \\*
\rule[-1.5ex]{0pt}{0pt}\hspace{2mm}Total PPs & 201,600 & 829,440 & 3,555,840 & 16,128,000 & 340,224,000 \\
\textbf{17} & & & & & \\*
\hspace{2mm}nPPs (ECs) & 0 & 209 (20) & 0 & 3,023 (201) & 50,608 (3,163) \\*
\rule[-1.5ex]{0pt}{0pt}\hspace{2mm}Total PPs & 0 & 966,416 & 0 & 13,978,352 & 234,011,392 \\
\textbf{19} & & & & & \\*
\hspace{2mm}nPPs (ECs) & 0 & 112 (12) & 864 (48) & 0 & 19,544 (1,094) \\*
\rule[-1.5ex]{0pt}{0pt}\hspace{2mm}Total PPs & 0 & 727,776 & 5,614,272 & 0 & 126,996,912 \\
\textbf{23} & & & & & \\*
\hspace{2mm}nPPs (ECs) & 0 & 89 (6) & 154 (7) & 3,092 (174) & 50,402 (2,291) \\*
\rule[-1.5ex]{0pt}{0pt}\hspace{2mm}Total PPs & 0 & 1,035,782 & 1,792,252 & 35,984,696 & 586,578,476 \\
\textbf{25} & & & & & \\*
\hspace{2mm}nPPs (ECs) & 0 & 45 (5) & 0 & 1,038 (38) & 401,280 (341) \\*
\rule[-1.5ex]{0pt}{0pt}\hspace{2mm}Total PPs & 0 & 675,000 & 0 & 15,570,000 & 240,768,000 \\
\textbf{27} & & & & & \\*
\hspace{2mm}nPPs (ECs) & 702 (1) & 14 (2) & 364 (6) & 29,550 (41) & 7,098 (95) \\*
\rule[-1.5ex]{0pt}{0pt}\hspace{2mm}Total PPs & 492,804 & 265,356 & 6,899,256 & 20,744,100 & 134,535,492 \\
\textbf{29} & & & & & \\*
\hspace{2mm}nPPs (ECs) & 0 & 0 & 32 (2) & 1,751 (67) & 1,568 (56) \\*
\rule[-1.5ex]{0pt}{0pt}\hspace{2mm}Total PPs & 0 & 0 & 753,536 & 41,232,548 & 36,923,264 \\
\textbf{31} & & & & & \\*
\hspace{2mm}nPPs (ECs) & 0 & 106 (6) & 30 (1) & 630 (22) & 0 \\*
\rule[-1.5ex]{0pt}{0pt}\hspace{2mm}Total PPs & 0 & 3,055,980 & 864,900 & 18,162,900 & 0 \\
\textbf{32} & & & & & \\*
\hspace{2mm}nPPs (ECs) & 1,024 (2) & 32 (2) & 19,624 (68) & 311 (7) & 410,720 (91) \\*
\rule[-1.5ex]{0pt}{0pt}\hspace{2mm}Total PPs & 1,015,808 & 1,015,808 & 19,467,008 & 9,872,384 & 407,434,240 \\
\textbf{37} & & & & & \\*
\hspace{2mm}nPPs (ECs) & 0 & 37 (3) & 0 & 0 & 216 (10) \\*
\rule[-1.5ex]{0pt}{0pt}\hspace{2mm}Total PPs & 0 & 1,823,508 & 0 & 0 & 10,645,344 \\
\textbf{41} & & & & & \\*
\hspace{2mm}nPPs (ECs) & 0 & 1 (1) & 0 & 331 (16) & 0 \\*
\rule[-1.5ex]{0pt}{0pt}\hspace{2mm}Total PPs & 0 & 67,240 & 0 & 22,256,440 & 0 \\
\textbf{43} & & & & & \\*
\hspace{2mm}nPPs (ECs) & 0 & 0 & 0 & 42 (2) & 98 (3) \\*
\rule[-1.5ex]{0pt}{0pt}\hspace{2mm}Total PPs & 0 & 0 & 0 & 3,261,636 & 7,610,484 \\
\textbf{47} & & & & & \\*
\hspace{2mm}nPPs (ECs) & 0 & 47 (3) & 0 & 116 (4) & 0 \\*
\rule[-1.5ex]{0pt}{0pt}\hspace{2mm}Total PPs & 0 & 4,775,858 & 0 & 11,787,224 & 0 \\
\textbf{49} & & & & & \\*
\hspace{2mm}nPPs (ECs) & 0 & 3,961 (10) & 0 & 96 (3) & 16 (1) \\*
\rule[-1.5ex]{0pt}{0pt}\hspace{2mm}Total PPs & 0 & 9,316,272 & 0 & 11,063,808 & 1,843,968 \\
\textbf{53} & & & & & \\*
\hspace{2mm}nPPs (ECs) & 0 & 53 (3) & 0 & 53 (2) & 0 \\*
\rule[-1.5ex]{0pt}{0pt}\hspace{2mm}Total PPs & 0 & 7,741,604 & 0 & 7,741,604 & 0 \\
\textbf{59} & & & & & \\*
\hspace{2mm}nPPs (ECs) & 0 & 59 (3) & 0 & 117 (4) & 0 \\*
\rule[-1.5ex]{0pt}{0pt}\hspace{2mm}Total PPs & 0 & 11,911,982 & 0 & 23,622,066 & 0 \\
\textbf{61} & & & & & \\*
\hspace{2mm}nPPs (ECs) & 0 & 61 (3) & 0 & 0 & 0 \\*
\rule[-1.5ex]{0pt}{0pt}\hspace{2mm}Total PPs & 0 & 13,618,860 & 0 & 0 & 0 \\
\textbf{64} & & & & & \\*
\hspace{2mm}nPPs (ECs) & 0 & 0 & 80,968 (214) & 0 & 21,120 (8) \\*
\rule[-1.5ex]{0pt}{0pt}\hspace{2mm}Total PPs & 0 & 0 & 326,462,976 & 0 & 85,155,840 \\
\textbf{67} & & & & & \\*
\hspace{2mm}nPPs (ECs) & 0 & 67 (3) & 0 & 0 & 0 \\*
\rule[-1.5ex]{0pt}{0pt}\hspace{2mm}Total PPs & 0 & 19,850,358 & 0 & 0 & 0 \\
\textbf{71} & & & & & \\*
\hspace{2mm}nPPs (ECs) & 0 & 0 & 0 & 71 (2) & 0 \\*
\rule[-1.5ex]{0pt}{0pt}\hspace{2mm}Total PPs & 0 & 0 & 0 & 25,053,770 & 0 \\
\textbf{73} & & & & & \\*
\hspace{2mm}nPPs (ECs) & 0 & 73 (3) & 0 & 0 & 0 \\*
\rule[-1.5ex]{0pt}{0pt}\hspace{2mm}Total PPs & 0 & 28,009,224 & 0 & 0 & 0 \\
\textbf{79} & & & & & \\*
\hspace{2mm}nPPs (ECs) & 0 & 79 (3) & 0 & 0 & \\*
\rule[-1.5ex]{0pt}{0pt}\hspace{2mm}Total PPs & 0 & 38,457,042 & 0 & 0 & \\
\textbf{81} & & & & & \\*
\hspace{2mm}nPPs (ECs) & 0 & 81 (3) & 0 & 471,891 (55) & 0 \\*
\rule[-1.5ex]{0pt}{0pt}\hspace{2mm}Total PPs & 0 & 42,515,280 & 0 & 3,057,853,680 & 0 \\
\textbf{83} & & & & & \\*
\hspace{2mm}nPPs (ECs) & 0 & 1 (1) & 0 & 83 (2) & \\*
\rule[-1.5ex]{0pt}{0pt}\hspace{2mm}Total PPs & 0 & 564,898 & 0 & 46,886,534 & \\
\textbf{89} & & & & & \\*
\hspace{2mm}nPPs (ECs) & 0 & 89 (3) & 0 & 89 (2) & \\*
\rule[-1.5ex]{0pt}{0pt}\hspace{2mm}Total PPs & 0 & 62,037,272 & 0 & 62,037,272 & \\
\textbf{97} & & & & & \\*
\hspace{2mm}nPPs (ECs) & 0 & 1 (1) & 0 & 0 & \\*
\rule[-1.5ex]{0pt}{0pt}\hspace{2mm}Total PPs & 0 & 903,264 & 0 & 0 & \\
\hline
\caption{Number of nPPs, Equivalence Classes, and Total number of PPs for $q \le 97$ and degree $d$, where $6 \le d \le 10$}
\label{tb1-Num-nPPs}
\end{longtable}

\begin{table}[htb]
\centering
\vspace*{4mm}
\begin{tabular}{|r|r|r|r||r|r|r|r|}
\hline
\bf  $n$ & \bf $D$ & \bf  $M(n,D) \geq$ & {\bf Previous} & \bf $n$ & \bf $D$ &  $M(n,D) \geq$ & {\bf Previous} \\
\hline
\hline
16 & 5 & 5,112,053,760 & 143,866,479 & 41	&	32	&	 22,392,560 	&	 1,565,096  \\
17 & 6 & 4,250,533,664 & 143,866,479 & 43	&	38	& 3,341,100 	&	 397,074 \\
18 & 9 & 72,480,384* &73,195,200 & 47	&	38	&	 21,442,716 	&	 103,776  \\ 
18 & 10 & 12,240,000* &1,269,376&47	&	40	&	 9,655,492 	&	 103,776 \\
19 & 8 & 2,565,261,288 & 143,866,479 &47	&	42	&	 4,879,634 	&	 103,776\\
24 & 20 & 24,288* & 23,782 &49	&	39	&	 23,341,648 	&	 207,552 	 \\
25 & 14 &6,766,500,000 & 143,866,479 & 49	&	40	&	 20,497,680 	&	 207,552  \\
25	&	15	&	 257,205,000 	&	 244,823,040 	& 49	&	42	&	 9,433,872 	&	 207,552 	\\
25	&	18	&	 867,000 	&	 279,818 	&53	&	44	&	 23,373,636 	&	 470,400 \\
25	&	20	&	 192,060 	&	 19,404 	&53	&	46	&	 15,632,032 	&	 148,824  				\\
27 & 16 & 2,990,448,396 &  143,866,479 & 53	&	48	&	 7,890,428 	&	 148,824	 \\
27	&	17	&	 163,459,296 	&	 9,313,200 	&59	&	50	&	 35,941,266 	&	 4,762,368 	\\
27	&	18	&	 289,233,804 	&	 9,313,200 	&59	&	52	&	 12,319,200 	&	 1,339,416 		\\
27	&	19	&	 8,179,204 	&	 1,326,000 	& 59	&	54	&	 407,218 	&	 205,320 		\\
27	&	21	&	 1,015,092 	&	 249,600 	&61	&	54	&	 13,622,520 	&	 410,640 			\\
27	&	22	&	 522,288 	&	 31,200 	&64	&	55	&	 332,236,800 	&	 262,080 	\\
29 & 18 & 1,093,475,740 &  143,866,479 & 64	&	59	&	 5,773,824 	&	 262,080 	 \\
29	&	19	&	 78,957,256 	&	 9,533,160 	&64	&	60	&	 5,515,776 	&	 262,080		\\
29	&	20	&	 42,033,992 	&	 9,533,160 	& 67	&	60	&	 39,705,138 	&	 524,160 	\\
31 & 20 & 407,138,190 &  143,866,479 & 67	&	62	&	 19,854,780 	&	 524,160  \\
31	&	22	&	 22,084,719 	&	 1,291,080 	&		71	&	62	&	 25,411,610 	&	 601,392 	\\
31	&	23	&	 3,921,819 	&	 1,291,080 	&73	&	66	&	 56,023,704 	&	 357,840 	\\
31	&	24	&	 3,056,919 	&	 372,992 	&		73	&	68	&	 28,014,480 	&	 357,840 	\\
32 & 21 & 460,134,208 & 143,866,479 &	79	&	72	&	 38,950,002 	&	 492,960 	\\
32	&	22	&	 440,194,048 	&	 1,291,080 	&81	&	72	&	 3,100,641,120 	&	 571,704 		\\
32	&	23	&	 32,759,808 	&	 1,291,080 	&	81	&	74	&	 42,787,440 	&	 571,704 		\\
32	&	24	&	 22,887,424 	&	 372,992 	& 83	&	74	&	 94,909,670 	&	 888,729 		\\
32	&	25	&	 3,420,416 	&	 372,992 	&	83	&	76	&	 48,023,136 	&	 571,704 	\\
32	&	26	&	 2,404,608 	&	 208,320 	&	83	&	78	&	 47,458,238 	&	 571,704 	\\
32	&	27	&	 1,388,800 	&	 372,992 	& 89	&	80	&	 125,476,464 	&	 1,062,720 			\\
32	&	29	&	 33,728 	&	 32,736 	&89	&	82	&	 63,439,192 	&	 704,880 	\\
37	&	27	&	 14,293,692 	&	 1,473,120 	&89	&	84	&	 1,401,920 	&	 704,880 	\\
37	&	30	&	 3,648,348 	&	 155,122 	&97	&	90	&	 88,529,184 	&	 912,576 	\\
37	&	32	&	 1,824,840 	&	 50,616 	&	97	&	92	&	 87,625,920 	&	 912,576 	\\
 
\hline
\end{tabular}
\caption{New lower bounds for $M(n,D)$. Note: * denotes values that were obtained by coset search \cite{bmms-19}, not by permutation polynomials. We include these for the sake of completeness.
}
\label{tbl-M(n,d)}
\end{table}

\begingroup  
\renewcommand{\arraystretch}{1.2} 

\begin{table}[htb]
\centering
\vspace*{4mm}
\begin{tabular}{|r|c|c|}
\hline
\bfseries  class & $a$ &  number of nPPs\\
\hline
$x^7$ & - & 1 \\
$x^7+ax^2$ & \{3,5\} & 4\\
$x^7+x^4+7x$ & - & 10 \\
$x^7+x^4+ax^2+3x$ & \{9,10\}& 20 \\
$x^7+x^5+3x^3+ax$ & \{2,4,9\} & 15
\\
$x^7+x^5+3x^3+5x^2+8x$ & - & 10 \\

$x^7+x^5+2ax^4+3x^3+4a^{-3}x^2+3x$ & \{1,2\} & 20 \\
$x^7+x^5+2ax^4+3x^3+5x^2$ & \{1,4\} & 20 \\
$x^7+x^5+4x^4+3x^3+9x^2+10x$ & - & 10 \\
$x^7+2x^5+5x^3+ax$ & \{2,3,4\} & 15 \\
$x^7+2x^5+5x^3+2x^2+x$ & - & 10 \\
$x^7+2x^5+ax^4+5x^3+8a^{-6}x^2+4x$ & \{1,2\} & 20\\
$x^7+2x^5+2x^4+5x^3+3x^2+8x$ & - & 10  \\
$x^7+2x^5+2ax^4+5x^3+10a^{-7}x^2+5a^4x$ & \{1,2\} & 20 \\
$x^7+2x^5+4ax^4+5x^3+5a^{-2}x^2+9a^{-1}x$ & \{1,2\} & 20 \\
$x^7+2x^5+5x^4+5x^3+8a^2x^2+5a^4x$ & \{1,2\} & 20 \\
 \hline
TOTAL &  & 225\\
\hline
\end{tabular}
\caption{nPP Equivalence Classes for $q=11$, degree 7, and primitive polynomial $x+4$.}
\label{tbl-EqCl-11-7}
\end{table}

\begin{table}[htb]
\centering
\vspace*{4mm}
\begin{tabular}{|r|c|}
\hline
\bfseries  class &  number of nPPs\\
\hline
$x^7$ & 1 \\
$x^7+x^5+4x^3+9x$ & 11\\
$x^7+x^5+x^4+21x^2+7x$ & 22 \\
$x^7+x^5+6x^4+2x^3+6x^2+7x$ & 22 \\
$x^7+x^5+11x^4+3x^3+10x^2+22x$ & 22 \\
$x^7+2x^5+6x^3+12x$ & 11 \\ 
 \hline
TOTAL &   89  \\
\hline
\end{tabular}
\caption{nPP Equivalence Classes for $q=23$, degree 7, and primitive polynomial $x+12$.}
\label{tbl-EqCl-23-7}
\end{table}

\begin{table}[htb]
\centering
\vspace*{4mm}
\begin{tabular}{|r|c|}
\hline
\bfseries  class &  number of nPPs\\
\hline
$x^7$ & 1 \\
$x^7+2x$ & 8\\
$x^7+x^5+x^3+13x$ & 12 \\
$x^7+2x^5+3x^3$ & 12 \\
$x^7+2x^5+3x^3+16x$ & 12 \\
 \hline
TOTAL &   45  \\
\hline
\end{tabular}
\caption{nPP Equivalence Classes for $q=25$, degree 7, and primitive polynomial $x^2+3x+3$.}
\label{tbl-EqCl-25-7}
\end{table}

\begin{table}[htb]
\centering
\vspace*{4mm}
\begin{tabular}{|r|c|c|}
\hline
\bfseries  class & $a$ & number of nPPs\\
\hline
$x^9+ax$ & \{1,7\} &  6 \\
$x^9+2ax^5+9a^2x$ & \{1,3\} &  12 \\
$x^9+x^7+2ax^5+22x^3+8a^{-1}x$ & \{1,2\} &  48  \\
$x^9+x^7+13x^5+4x^3+ax$ & \{5,6,22\} &  72  \\
$x^9+2x^7+10ax^3+23a^{-1}x$ & \{1,7\} & 24 \\
$x^9+2x^7+4x^5+4a^5x^3+21a^{-7}x$ & \{1,2\}& 48  \\
$x^9+2x^7+12a^9x^5+4a^{18}x^3+17a^{-12}x$ & \{1,2\} &  36  \\
$x^9+2x^7+3x^5$ & - & 12  \\
$x^9+2x^7+9x^5+5x$ & - & 12 \\

$x^9+x^7+9x^5+16a^3x^3+10a^2x$  & \{1,3\} &  48 \\
$x^9+x^7+3a^5x^5+16a^2x^3+10a^2x$ & \{1,2\} & 48  \\
$x^9+x^7+14ax^5+3a^4x^4+12x^3+11a^2x^2$ & \{1,2\} & 96  \\
$x^9+2x^7+5a^{11}x^5+x^3+8a^{-7}x^4+18a^{6}x^3+ \quad \quad $ &  &  \\
$+17a^6x^3+17a^{-1}x^2+9a^{-7}x$ & \{1,2\}  & 96  \\
$x^9+x^7+3a^6x^5+9a^{-2}x^4+20a^{-8}x^3+9a^{-8}x^2+6a^7x$ & \{1,2\} & 96 \\ 
$x^9+2x^7+9x^5+4a^6x^4+18a^{-14}x^3+18a^{-17}x^2+8a^{-4}x$ & \{1,2\} & 96 \\

$x^9+2x^7+2a^5x^5+16a^{-6}x^3+a^{14}x$ & \{1,2\} & 96
\\
$x^9+2x^7+11x^5+3a^2x^48a^9x^3+23a^{-1}x^2+9x$ & \{1,2\} & 96
\\
 $x^9+x^7+7x^5+4x^2+x$ & - &  24 \\
$x^9+2x^7+9x^5+16a^4x^3+5a^{12}x$ & \{1,2\} &  24 \\
 $x^9+2x^7+2x^5+x^4+19x^3+14x^2+23x$ & - &  48\\ 
 $x^9+2x^7+4x^5+11x^4+8x^3+9x$ & - &  48
 \\
 \hline
TOTAL &  &   1038\\
\hline
\end{tabular}
\caption{nPP Equivalence Classes for $q=25$, degree 9, and primitive polynomial $x^2+3x+3$.}
\label{tbl-EqCl-25-9}
\end{table}

\begin{table}[htb]
\centering
\vspace*{4mm}
\begin{tabular}{|r|c|c|}
\hline
\bfseries  class & $\alpha(x)$ & number of nPPs\\
\hline
$x^{10}+2x^8+\alpha(x)$ & $\{3x^6+4x^5+4x^4+x^3+22x^2+21x, $  & 3600 \\ 
& $10x^6+4x^5+20x^4+6x^2+23x,$ & \\
& $12x^6+10x^5+24x^4+6x^3+7x^2+16x\}$ &  \\

 \hline
TOTAL &  & 3600  \\
\hline
\end{tabular}
\caption{nPP Equivalence Classes for $q=25$, degree 10, and primitive polynomial $x^2+3x+3$.}
\label{tbl-EqCl-25-10}
\end{table}

\begin{table}[htb]
\centering
\vspace*{4mm}
\begin{tabular}{|r|c|}
\hline
\bfseries  class  & number of nPPs\\
\hline
$x^8+x^6+14x^4+x^3+14x^2$ &    26 \\
$x^8+x^6+x^5+14x^4+x^3+14x^2$ &  26 \\
$x^8+x^6+2x^5+4x^4+15x^3+14x^2+7x$ & 78  \\
$x^8+x^6+2x^5+24x^4ax+9x^3+14x^2+11x$& 78  \\
$x^8+x^6+8x^5+11x^4+16x^3+14x^2+6x$ & 78 \\
$x^8+2x^6+6x^5+26x^4+6x^3+17x^2+2x$ &  78   \\
\hline
TOTAL &    364\\
\hline
\end{tabular}
\caption{nPP Equivalence Classes for $q=27$, degree 8, and primitive polynomial $x^3+2x+1$.}
\label{tbl-EqCl-27-8}
\end{table}

\begin{table}[htb]
\centering
\vspace*{4mm}
\begin{tabular}{|r|c|c|}
\hline
\bfseries  class & $a$ & number of nPPs\\
\hline
$x^9$ & - &  1 \\
$x^9+x$ & - &  13 \\
$x^9+x^3+ax$ & \{0,2,3,5,8,9,14\} &  221  \\
$x^9+2x^3+ax$ & \{1,5,6,11,12,24\} &  182  \\
$x^9+x^5+x^3+x$ & - & 351 \\
$x^9+x^5+14x^3+ax$ & \{0,1\}& 702  \\
$x^9+2x^5+3x$& - &  351  \\
$x^9+2x^5+6a^2x^3+24ax$ &\{1,19\} & 1404  \\
$x^9+x^6+ax^5+x^3+14ax^2+a^2x$ & \{1,8,9,15\} & 7020 \\
$x^9+x^7+x^3+ax$  &\{1,14\} &  702  \\
$x^9+x^7+3ax^3+12a^{-1}x$ & $\{1,12\}$ &  1404 \\
$x^9+x^7+a^2x^4+x^3+14a^2x^2+14ax$ &\{1,3\}  &  2808 \\
$x^9+x^7+3ax^4+x^3+16ax^2+19a^{10}x$ &\{1,6\}  &  4212 \\
$x^9+x^7+2x^4+a^{16}x^3+14ax^2+a^5x$ &\{1,10\} &4212 \\ 
$x^9+2x^7+x$ & - & 351 \\
$x^9+2x^7+ax^3$ & \{4,17\} & 702 \\
 $x^9+2x^7+2x^4+11x^3+x^2+14x$ & - &  2106 \\
 $x^9+2x^7+10a^5x^3+17ax^2+11ax$ & \{1,5\} &  2808\\
 \hline
TOTAL &  &   29550\\
\hline
\end{tabular}
\caption{nPP Equivalence Classes for $q=27$, degree 9, and primitive polynomial $x^3+2x+1$.}
\label{tbl-EqCl-27-9}
\end{table}

\begin{table}[htb]
\centering
\vspace*{4mm}
\begin{tabular}{|r|c|c|}
\hline
\bfseries  class & $a$ & number of nPPs\\
\hline
$x^8$ & - &  1 \\
$x^8+x^2+ax$ &\{1,12,16\}  &  341 \\
$x^8+x^4+ax$ &\{1,2,8\}  & 341  \\
$x^8+x^4+x^2+ax$ & \{0,2,6\} & 341  \\
$x^8+x^4+2x^2+ax$ & \{0,3,6,7,8,13,14,16,18,19,24,25,31\}    & 2015 \\
$x^8+x^4+4x^2+ax$ & \{0,13,15,20,28,29,30\}& 1085   \\
$x^8+x^4+6x^2+ax$&\{4,7,11,12,15,20,23,24,28,29,30\} & 2015 \\
$x^8+x^4+8x^2+ax$&\{1,2,5,9,12,13,19,25,28\} & 1395 \\
$x^8+x^4+12x^2+ax$&\{3,4,6,11,20,29,30\} & 1085 \\
$x^8+x^4+16x^2+ax$&\{1,5,10,12,15,17,18,30\} & 1240 \\
$x^8+x^6+4x^5+x^4+23x^3+2x^2+23x$& - & 4960 \\
$x^8+x^6+6x^5+7x^3+2x^2$& - & 4960 \\
\hline
TOTAL &  &   19624\\
\hline
\end{tabular}
\caption{nPP Equivalence Classes for $q=32$, degree 8, and primitive polynomial $x^5+x^3+x^2+x+1$.}
\label{tbl-EqCl-32-8}
\end{table}

\begin{table}[htb]
\centering
\vspace*{4mm}
\begin{tabular}{|r|c|}
\hline
\bfseries  class  & number of nPPs\\
\hline
$x^9$  &  1 \\
$x^9+x^3+x$  &  31 \\
$x^9+x^6+x^2$  & 31  \\
$x^9+x^6+x^3$  & 31  \\
$x^9+x^7+x$    & 31 \\
$x^9+x^7+x^6+x^4+x^3$ & 31   \\
$x^9+x^7+16x^6+22x^5+16x^4+30x^3+17x^2+5x$ & 155 \\

\hline
TOTAL &    311\\
\hline
\end{tabular}
\caption{nPP Equivalence Classes for $q=32$, degree 9, and primitive polynomial $x^5+x^3+x^2+x+1$.}
\label{tbl-EqCl-32-9}
\end{table}

\begin{table}[htb]
\centering
\vspace*{4mm}
\begin{tabular}{|r|c|}
\hline
\bfseries  class & number of nPPs\\
\hline
$x^7$ &  1 \\
$x^7+x^5+26x^3+5x$ &   18 \\
$x^7+2x^5+28x^3+8x$ &   18 \\
\hline
TOTAL &  37  \\
\hline
\end{tabular}
\caption{nPP Equivalence Classes for $q=37$, degree 7, and primitive polynomial $x+13$.}
\label{tbl-EqCl-37-7}
\end{table}

\begin{table}[htb]
\centering
\vspace*{4mm}
\begin{tabular}{|r|c|c|}
\hline
\bfseries  class & $a$ & number of nPPs\\
\hline
$x^9$ & - &  1 \\
$x^9+2ax^5+3a^2x$ &\{1,2,3\}  &  30 \\
$x^9+x^6+6x^3$ & - &  40  \\
$x^9+x^7+4a^9x^5+36a^{19}x^3+33a^{-1}x$ & \{1,2\} &  40  \\
$x^9+x^7+14a^{25}x^5+28a^{-2}x^3+5a^4x$ & \{1,2\} & 40 \\
$x^9+ax^7+19a^{-8}x^6+5a^{11}x^5+3a^{-7}x^4+ \quad \quad $ & &\\
$+16a^{14}x^3+15a^7x^2+39a^{-10}x$ & \{1,2\} & 80   \\
$x^9+2x^7+15a^{16}x^5+2a^{-3}x^3+31a^{22}x$&\{1,2\} & 40 \\
$x^9+2x^7+33ax^5+18a^{22}x^15+15a^{22}x$&\{1,2\} & 40 \\
$x^9+2x^7+39x^5+25x^3+35x$& - & 20 \\
\hline
TOTAL &  &   331\\
\hline
\end{tabular}
\caption{nPP Equivalence Classes for $q=41$, degree 9, and primitive polynomial $x+17$.}
\label{tbl-EqCl-41-9}
\end{table}

\begin{table}[htb]
\centering
\vspace*{4mm}
\begin{tabular}{|r|c|}
\hline
\bfseries  class & number of nPPs\\
\hline
$x^9+2x^7+28x^5+25x^3+2x$ &  21 \\
$x^9+2x^7+39x^5+33x^3+30x$ &   21 \\
\hline
TOTAL &  42  \\
\hline
\end{tabular}
\caption{nPP Equivalence Classes for $q=43$, degree 9, and primitive polynomial $x+13$.}
\label{tbl-EqCl-43-9}
\end{table}

\begin{table}[htb]
\centering
\vspace*{4mm}
\begin{tabular}{|r|c|}
\hline
\bfseries  class  & number of nPPs\\
\hline
$x^{10}+x^7+28x^4+42x$ &  14 \\
$x^{10}+x^8+19x^7+x^6+33x^5+29x^4+13x^3+x^2+27x$ &  42 \\
$x^{10}+2x^8+12x^7+32x^6+32x^5+10x^4+39x^3+15x^2+10x$ &  42 \\
\hline
TOTAL & 98  \\
\hline
\end{tabular}
\caption{nPP Equivalence Classes for $q=43$, degree 10, and primitive polynomial $x+13$.}
\label{tbl-EqCl-43-10}
\end{table}

\begin{table}[htb]
\centering
\vspace*{4mm}
\begin{tabular}{|r|c|}
\hline
\bfseries  class  & number of nPPs\\
\hline
$x^9$ &  1 \\
$x^9+x^6+45x^3$ &  46 \\
$x^9+x^6+13x^5+46x^4+44x^3+39x^2+19x$ &  46  \\
$x^9+2x^7+13x^5+32x^3+22x$ &  23  \\
\hline
TOTAL &  116   \\
\hline
\end{tabular}
\caption{nPP Equivalence Classes for $q=47$, degree 9, and primitive polynomial $x+12$.}
\label{tb-EqCl-47-9}
\end{table}

\begin{table}[htb]
\centering
\vspace*{4mm}
\begin{tabular}{|r|c|}
\hline
\bfseries  class & number of nPPs\\
\hline
$x^9+3x^5+10x$ &  24 \\
$x^9+2x^7+3x^5+4x^3$ &  24 \\
$x^9+2x^7+21x^5+26x^3+36x$ &  48
\\
\hline
TOTAL &  96  \\
\hline
\end{tabular}
\caption{nPP Equivalence Classes for $q=49$, degree 9, and primitive polynomial $x^2+6x+3$.}
\label{tbl-EqCl-49-9}
\end{table}

\begin{table}[htb]
\centering
\vspace*{4mm}
\begin{tabular}{|r|c|}
\hline
\bfseries  class & number of nPPs\\
\hline
$x^{10}+x^7+9
x$ &  16
\\
\hline
TOTAL & 16  \\
\hline
\end{tabular}
\caption{nPP Equivalence Classes for $q=49$, degree 10, and primitive polynomial $x^2+6x+3$.}
\label{tbl-EqCl-49-10}
\end{table}

\begin{table}[htb]
\centering
\vspace*{4mm}
\begin{tabular}{|r|c|}
\hline
\bfseries  class  & number of nPPs\\
\hline
$x^9$ &  1 \\
$x^9+x^6+3x^3$ &  70 \\
\hline
TOTAL & 71  \\
\hline
\end{tabular}
\caption{nPP Equivalence Classes for $q=71$, degree 9, and primitive polynomial $x+38$.}
\label{tbl-EqCl-71-9}
\end{table}

\begin{table}[htb]
\centering
\vspace*{4mm}
\begin{tabular}{|r|c|}
\hline
\bfseries  class  & number of nPPs\\
\hline
$x^7$ &  1 \\
$x^7+x^5+44x^3+31x$ &  36 \\
$x^7+2x^5+46x^3+34x$ &   36 \\
\hline
TOTAL & 73  \\
\hline
\end{tabular}
\caption{nPP Equivalence Classes for $q=73$, degree 7, and primitive polynomial $x+20$.}
\label{tbl-EqCl-73-7}
\end{table}

\begin{table}[htb]
\centering
\vspace*{4mm}
\begin{tabular}{|r|c|}
\hline
\bfseries  class  & number of nPPs\\
\hline
$x^7$  &  1 \\
$x^7+x^5+72x^3+75x$ & 39 \\
$x^7+2x^5+74x^3+78x$ &  39  \\
\hline
TOTAL &   79\\
\hline
\end{tabular}
\caption{nPP Equivalence Classes for $q=79$, degree 7, and primitive polynomial $x+25$.}
\label{tbl-EqCl-79-7}
\end{table}

\begin{table}[htb]
\centering
\vspace*{4mm}
\begin{tabular}{|r|c|c|}
\hline
\bfseries  class & $a$ & number of nPPs\\
\hline
$x^9+ax$ & $\{0,2,3,5,6\}$ &  71 \\
$x^9+ax^3+ax$ & $\{2,3,5,8,12,17,23, \quad \quad $  & 1880 \\
  & $\quad 24,26,27,41,45,51\}$  & \\
$x^9+2x^3+ax$ & $\{0,3,5,6,7,9,10,12, \quad \quad $ &  1760  \\
  & $ \quad 14,19,29,30,33,39\}$  &  \\
$x^9+x^2ax^5+3a^2x$ & $\{1,2\}$ &  4860  \\
$x^9+x^6+8ax^5+x^3+48ax^2+15a^2x$ & $\{1,2\}$ & 38880 \\
$x^9+x^6+15ax^5+x^3+55ax^2+29a^2x$ & $\{1,4\}$ & 51840 \\
$x^9+2x^7+7a^2x^4+4x^3+48a^2x^2+67a^{-3}x$ & $\{1,2\}$ & 51840 \\
$x^9+x^6+24ax^5+x^3+64ax^2+47a^2x$ & $\{1,3\}$ & 51840 \\
$x^9+x^6+42ax^5+x^3+2ax^2+3a^2x$ & $\{1,4\}$ & 51840 \\
$x^9+2x^7+4x^3+5x$ &  & 3240 \\
$x^9+2x^7+ax^4+4x^3+42ax^2+11a^{-5}x$ & $\{1,2\}$ & 51840 \\
$x^9+2x^7+4a^2x^4+4x^3+45a^2x^2+33a^{-17}x$ & $\{1,2\}$ & 51840 \\
$x^9+2x^7+10a^2x^4+4x^3+51a^2x^2+57a^{-3}x$ & $\{1,2\}$ & 51840 \\ 
$x^9+2x^7+15a^3x^4+4x^3+56a^3x^2+21a^7x$ & $\{1,3\}$ & 51840 \\
$x^9+2x^7+3x^5+4x^3+5x$ & & 3240 \\
$x^9+2x^7+43x^5+5x$ & & 3240 \\

\hline
TOTAL &  & 471891  \\
\hline
\end{tabular}
\caption{nPP Equivalence Classes for $q=81$, degree 9, and primitive polynomial $x^4+2x^3+2$.}
\label{tbl-EqCl-81-9}
\end{table}

\begin{table}[htb]
\centering
\vspace*{4mm}
\begin{tabular}{|r|c|}
\hline
\bfseries  class  & number of nPPs\\
\hline
$x^7$  &  1 \\
$x^7+x^5+2x^3+59x$ & 44 \\
$x^7+2x^5+4x^3+62x$ &  44  \\
\hline
TOTAL &   89\\
\hline
\end{tabular}
\caption{nPP Equivalence Classes for $q=89$, degree 7, and primitive polynomial $x+76$.}
\label{tbl-EqCl-89-7}
\end{table}

\begin{table}[htb]
\centering
\vspace*{4mm}
\begin{tabular}{|r|r|r|r|r|}
\hline
\bfseries  $q$ & $d$ & class  & primitive polynomial & number of nPPs\\
\hline
32  &  7 & $x^7$ & $x^5+x^3+x^2+x$ & 1\\
& & $x^7+x^5+x$ & & 31  \\
\hline
41 & 7 & $x^7$ & $x+17$ & 1 \\
\hline
53 & 7 & $x^7$ & $x+51$ & 1 \\
& & $x^7+x^5+40x^3+25x$ & & 26 \\
& & $x^7+2x^5+42x^3+28x$ & & 26 \\
\hline
59& 7 & $x^7$ & $x+5$ & 1 \\
& & $x^7+x^5+32x^3+53x$ & & 29 \\
& & $x^7+2x^5+34x^3+56x$ & & 29 \\
& 9 &$x^9$ & & 1 \\
& & $x^9+x^6+41x^3$ & & 59 \\
& & $x^9+x^7+45x^5+22x^3+21x$ & & 29 \\
& & $x^9+x^7+29x^5+4x^3+49x$ & & 29 \\
\hline
61 & 7 & $x^7$ & $x+10$ & 1 \\
& & $x^7+x^5+25x^3+15x$ & & 30 \\
& & $x^7+2x^5+27x^3+18x$ & & 30 \\
\hline
67 & 7 & $x^7$ & $x+17$ & 1 \\
& & $x^7+x^5+45x^3+57x$& & 33 \\
& & $x^7+2x^5+47x^3+60x$ & & 33\\
\hline
83 & 9 & $x^9$ & $x+61$ & 1 \\
& & $x^9+x^6+67x^3$ & & 82 \\
\hline
89 & 9 & $x^9$ & $x+76$ & 1 \\
& & $x^9+x^6+66x^3$ & & 88 \\
\hline
97&7&$x^7$&$x+90$&1 \\
\hline
\end{tabular}
\caption{Number of nPPs in Equivalence Classes for prime powers $q$ and degree $d$, for given primitive polynomial.}
\label{tbl-EqCl-many}
\end{table}

\endgroup

\end{document}